\documentclass[copyright,creativecommons]{eptcs}

\usepackage[utf8x]{inputenc}
\usepackage[T1]{fontenc}
\usepackage{xspace}
\usepackage{tikz}
\usepackage{amssymb}
\usepackage{eurosym}
\usepackage{mathtools}
\usepackage{amsthm}
\usepackage[normalem]{ulem}
\usepackage{subfig}
\usepackage{booktabs}
\usepackage{array}

\title{Compositionality, Decompositionality and Refinement in Input/Output Conformance Testing -- Technical Report}
\author{Lars Luthmann\thanks{This work has been supported by the German Research Foundation (DFG) in the Priority Programme SPP 1593: Design For Future -- Managed Software Evolution (LO 2198/2-1).}
\institute{Real-Time Systems Lab\\TU Darmstadt, Germany}
\email{lars.luthmann@es.tu-darmstadt.de}
\and
Stephan Mennicke\thanks{This work has been supported by the German Research Foundation (DFG), grant GO-671/6-2.}
\institute{Institute for Programming and Reactive Systems\\TU Braunschweig, Germany}
\email{mennicke@ips.cs.tu-bs.de}
\and
Malte Lochau$^*$
\institute{Real-Time Systems Lab\\TU Darmstadt, Germany}
\email{malte.lochau@es.tu-darmstadt.de}
}
\def\titlerunning{Compositionality, Decompositionality and Refinement in I/O Conformance Testing}
\def\authorrunning{L. Luthmann, S. Mennicke, and M. Lochau}

\newcommand{\eg}{\mbox{e.\,g.,}\xspace}
\newcommand{\etal}{\mbox{et al.}\xspace}
\newcommand{\ie}{\mbox{i.\,e.,}\xspace}

\newcommand{\cf}{\mbox{cf.}\xspace}

\newcommand{\ioco}{\textbf{ioco}\xspace}
\newcommand{\mioco}{\ensuremath{\mathmioco}\xspace}
\newcommand{\miocophi}{\ensuremath{\mathmiocophi}\xspace}
\newcommand{\miocod}{\ensuremath{\mathmiocod}\xspace}

\newcommand{\miaallowed}{\mbox{MIA\ensuremath{_A}}\xspace}
\newcommand{\miaforbidden}{\mbox{MIA\ensuremath{_\Phi}}\xspace}
\newcommand{\miadc}{\mbox{MIA\ensuremath{_\textit{DC}}}\xspace}

\newcommand{\transitionarrow}{\longrightarrow}
\newcommand{\Transitionarrow}{\Longrightarrow}
\newcommand{\mustarrow}{\transitionarrow_\Box}

\newcommand{\mayarrow}{\transitionarrow_\Diamond}

\newcommand{\oversetmust}[1]{\overset{#1}{\transitionarrow}_\Box}
\newcommand{\Oversetmust}[1]{\overset{#1}{\Transitionarrow}_\Box}
\newcommand{\oversetgamma}[1]{\overset{#1}{\transitionarrow}_\gamma}
\newcommand{\Oversetgamma}[1]{\overset{#1}{\Transitionarrow}_\gamma}
\newcommand{\oversetmay}[1]{\overset{#1}{\transitionarrow}_\Diamond}
\newcommand{\Oversetmay}[1]{\overset{#1}{\Transitionarrow}_\Diamond}
\newcommand{\transrel}[1]{\overset{#1}{\transitionarrow}}
\newcommand{\Transrel}[1]{\overset{#1}{\Transitionarrow}}

\newcommand{\mustarrowsub}[1]{\transitionarrow_{\Box#1}}

\newcommand{\mayarrowsub}[1]{\transitionarrow_{\Diamond#1}}

\newcommand{\oversetmustsub}[2]{\overset{#1}{\transitionarrow}_{\Box#2}}

\newcommand{\oversetmaysub}[2]{\overset{#1}{\transitionarrow}_{\Diamond#2}}

\newcommand\mayinit{\ensuremath{\textit{init}_\Diamond}\xspace}
\newcommand\mustinit{\ensuremath{\textit{init}_\Box}\xspace}
\newcommand\mayafter{\ensuremath{\after_\Diamond}\xspace}
\newcommand\mustafter{\ensuremath{\after_\Box}\xspace}
\newcommand\mayout{\ensuremath{\textit{Out}_\Diamond}\xspace}
\newcommand\mustout{\ensuremath{\textit{Out}_\Box}\xspace}
\newcommand\maystraces{\ensuremath{\textit{Straces}_\Diamond}\xspace}

\newcommand\mustafterd{\ensuremath{\after_{\Box D}}\xspace}
\newcommand\straces{\ensuremath{\textit{Straces}}\xspace}
\newcommand\out{\ensuremath{\textit{Out}}\xspace}
\newcommand\init{\ensuremath{\textit{init}}\xspace}

\newcommand{\miarefforbidden}{\ensuremath{\sqsubseteq_\Phi}\xspace}

\newcommand{\otimesforbidden}{\ensuremath{\otimes_\Phi}}

\newcommand{\otimesforbiddenhiding}{\ensuremath{\otimes_\Phi^H}}
\newcommand{\parallelallowed}{\ensuremath{\parallel_A}}
\newcommand{\parallelforbidden}{\ensuremath{\parallel_\Phi}}
\newcommand{\hidingallowed}{\ensuremath{\mid_A}}
\newcommand{\hidingforbidden}{\ensuremath{\mid_\Phi}}

\newcommand{\sslash}{\mathbin{/\mkern-6mu/}}
\newcommand{\sslashallowed}{\ensuremath{\mathbin{\sslash_{\mkern-8muA}}}}
\newcommand{\sslashforbidden}{\ensuremath{\mathbin{\sslash_{\mkern-6mu\Phi}}}}

\DeclareMathOperator{\after}{\mathbf{after}\xspace}
\DeclareMathOperator{\mathmiocophi}{\mathbf{mioco}_\Phi}
\DeclareMathOperator{\mathmiocod}{\mathbf{mioco}_D}
\DeclareMathOperator{\mathmioco}{\mathbf{mioco}_{\mathsf{MIA}}}
\DeclareMathOperator{\mathioco}{\mathbf{ioco}}

\DeclareMathOperator{\mathhiding}{\mathbf{hiding}}

\newcolumntype{L}[1]{>{\raggedright\let\newline\\\arraybackslash\hspace{0pt}}m{#1}}
\newcolumntype{C}[1]{>{\centering\let\newline\\\arraybackslash\hspace{0pt}}m{#1}}
\newcolumntype{R}[1]{>{\raggedleft\let\newline\\\arraybackslash\hspace{0pt}}m{#1}}

\usetikzlibrary{arrows,positioning,calc,automata}
\tikzset{
	every node/.style={},
	descr/.style={fill=white, inner sep=2pt},
	state/.style={circle,draw,fill=black,inner sep=.2em},
	initial text={},
	every initial by arrow/.style={*->,>=stealth'},
	emptystate/.style={inner sep=.2em},
	labeledcircledstate/.style={circle,draw,inner sep=.2em}, 
	labeledstate/.style={inner sep=.2em}, 
	may/.style={->,>=stealth',dashed},
	must/.style={->,>=stealth'},
	empty/.style={->,>=stealth',transparent},
	labeledroundedstate/.style={rounded corners,draw,inner sep=.2em,align=center,text height=.6em,text width=1.3em,text depth=.2em},
	labeledroundedwidestate/.style={rounded corners,draw,inner sep=.2em,align=center,text height=.6em,text width=2.6em,text depth=.2em},
	labeledroundedemptystate/.style={rounded corners,inner sep=.2em,transparent,text height=.6em,text width=1.3em,text depth=.2em}
}

\newcommand\scalefactor{.8}

\newtheorem{definition}{Definition}
\newtheorem{theorem}{Theorem}
\newtheorem{lemma}{Lemma}
\newtheorem{corollary}{Corollary}

\begin{document}

\maketitle


\begin{abstract}
We propose an input/output conformance testing theory utilizing Modal Interface Automata with Input Refusals (IR-MIA) as novel behavioral formalism for both the specification and the implementation under test. A modal refinement relation on IR-MIA allows distinguishing between obligatory and allowed output behaviors, as well as between implicitly underspecified and explicitly forbidden input behaviors. The theory therefore supports positive and negative conformance testing with optimistic and pessimistic environmental assumptions. We further show that the resulting conformance relation on IR-MIA, called modal-irioco, enjoys many desirable properties concerning component-based behaviors. First, modal-irioco is preserved under modal refinement and constitutes a preorder under certain restrictions which can be ensured by a canonical input completion for IR-MIA. Second, under the same restrictions, modal-irioco is compositional with respect to parallel composition of IR-MIA with multi-cast and hiding. Finally, the quotient operator on IR-MIA, as the inverse to parallel composition, facilitates decompositionality in conformance testing to solve the unknown-component problem.
\end{abstract}

%
\section{Introduction}\label{sec:introduction}
Formal approaches to model-based testing of
component-based systems define notions of
\emph{behavioral conformance} between a \emph{specification}
and a (black-box) \emph{implementation (under test)}, both
usually given as (variations of) 
labeled transition systems (LTS).
Existing notions of behavioral conformance may be 
categorized into two research directions.
\emph{Extensional} approaches define
\emph{observational equivalences}, requiring
that no observer process (tester) is ever able to distinguish
behaviors shown by the implementation 
from those allowed by the specification~\cite{Nicola1987}.
In contrast, \emph{intensional} approaches rely on I/O labeled
transition systems (IOLTS) from which \emph{test cases} are derived as sequences
of controllable input and observable output actions,
to establish an \emph{alternating simulation 
relation} on IOLTS~\cite{Veanes2012,Gregorio2013}.
One of the most prominent conformance testing theories, 
initially introduced by Tretmans in~\cite{Tretmans1996}, combines
both views on formal conformance testing into an input/output 
conformance (\textbf{ioco}) relation on IOLTS.
Although many formal properties of, and extensions to, \textbf{ioco} have been
intensively investigated, \textbf{ioco} still suffers 
several essential weaknesses.
\begin{itemize}
\item The \textbf{ioco} relation permits \emph{underspecification} by means
of (1) unspecified input behaviors and (2)~non-deterministic input/output behaviors.
But, concerning (1), \textbf{ioco} is limited to \emph{positive} testing (\ie unspecified
inputs may be implemented arbitrarily) thus implicitly
relying on optimistic environmental assumptions.
Also supporting \emph{negative} testing in a pessimistic setting, however, would require
a distinction between \emph{critical} 
and \emph{uncritical} unintended input behaviors.
Concerning (2), \textbf{ioco} requires the implementation to exhibit
\emph{at most} output behaviors permitted by the specification.
In addition, the notion of \emph{quiescence} (\ie observable absence of any outputs) 
enforces implementations to show \emph{at least one} specified output behavior (if any).
Apart from that, no explicit distinction between \emph{obligatory} 
and \emph{allowed} output behaviors is expressible in IOLTS.	 
\item \textbf{ioco} imposes a special kind of \emph{alternating simulation} between
specification and implementation which is, in general, not a preorder,
although being a crucial property for testing relations on LTS~\cite{Nicola1995}.
\item \textbf{ioco} lacks a unified theory for input/output 
conformance testing in the face of component-based behaviors
being compatible with potential solutions for the aforementioned
weaknesses.
\end{itemize}
As all these weaknesses mainly stem from the
limited expressiveness of IOLTS as behavioral formalism,
we propose \emph{Modal Interface Automata with Input Refusals (IR-MIA)}
as a new model for input/output conformance testing for
both the specification and the implementation 
under test.
IR-MIA adopt Modal Interface Automata (MIA)~\cite{Bujtor2015a}, 
which combine concepts of Interface Automata~\cite{Alfaro2001}
(\ie I/O automata permitting underspecified input behaviors)
and (I/O-labeled) Modal Transitions 
Systems~\cite{Larsen2007,Bauer2010,Raclet2011}
(\ie LTS with distinct mandatory and optional transition relations).
In particular, we exploit enhanced versions of MIA
supporting both optimistic and pessimistic 
environmental assumptions~\cite{Luettgen2014} and
non-deterministic input/output behaviors~\cite{Bujtor2015a}.
For the latter, we have to re-interpret the universal state of MIA,
simulating every possible behavior, as 
\emph{failure state} to serve as target for those unintended, yet critical
input behaviors to be \emph{refused} by the implementation~\cite{Phillips1987}.
Modal refinement of IR-MIA therefore allows distinguishing between 
obligatory and allowed output behaviors, as well as between 
implicitly underspecified and explicitly forbidden input behaviors.

The resulting testing theory on IR-MIA
unifies positive and negative conformance testing 
with optimistic and pessimistic environmental assumptions.
We further prove that the corresponding modal I/O conformance 
relation on IR-MIA, called \textbf{modal-irioco}, 
exhibits essential properties, especially 
with respect to component-based systems testing.
\begin{itemize}
	\item \textbf{modal-irioco} is preserved under modal refinement
and constitutes a preorder under certain restrictions
which can be obtained by a canonical input completion~\cite{Rensink2007}. 
	\item \textbf{modal-irioco} is compositional with respect to parallel composition 
	of IR-MIA with multi-cast and hiding~\cite{Bujtor2015a}. 
	\item \textbf{modal-irioco} allows for decomposition of conformance testing,
	thus supporting environmental synthesis for component-based testing in contexts~\cite{Noroozi2013,Daca2014}, also 
	known as the \emph{unknown-component problem}~\cite{Villa2011}.
	To this end, we adapt the MIA quotient operator to IR-MIA, 
	serving as the inverse to parallel composition.
\end{itemize}
The remainder of this paper
is organized as follows.
In Sect.~\ref{sec:preliminaries}, we revisit the foundations of \textbf{ioco} testing.
In Sect.~\ref{sec:mioco}, we introduce IR-MIA and modal refinement on IR-MIA and, thereupon,
define \textbf{modal-irioco}, provide a correctness proof and discuss necessary restrictions
to obtain a preorder.
Our main results concerning compositionality and decompositionality
of \textbf{modal-irioco} are presented in Sect.~\ref{sec:compositionality} and
Sect.~\ref{sec:decompositionality}, respectively.
In Sect.~\ref{sec:casestudy}, we present a case study explaining the introduced concepts with a real-world example.
In Sect.~\ref{sec:related}, we discuss related work and in Sect.~\ref{sec:conclusion}, we
conclude the paper.
Please note that all proofs may be found in Appendix~\ref{sec:proofs}.
%
\section{Preliminaries}\label{sec:preliminaries}
The \textbf{ioco} testing theory
relies on I/O-labeled transition systems
(IOLTS) as behavioral formalism~\cite{Tretmans1996}.
An IOLTS $(Q,I,O,\longrightarrow)$ specifies the \emph{externally visible} behaviors
of a system or component by means of a transition
relation $\longrightarrow\,\subseteq Q\times (I\cup O \cup \{\tau\}) \times Q$
on a set of states $Q$.
The set of transition labels $\textit{A}=I\cup O$
consists of two disjoint subsets: set $I$ of
externally controllable/internally observable \emph{input} actions,
and set $O$ of internally controllable/externally observable \emph{output} actions.
In figures, we use prefix \textit{?} to mark input
actions and prefix \textit{!} for output actions, respectively.
In addition, transitions labeled with \emph{internal}
actions $\tau\not\in (I\cup O)$ denote silent moves,
neither being externally controllable, nor observable.
We write $A^\tau=A\cup\{\tau\}$, and by $q\transrel{\alpha}q'$ we denote that $(q,\alpha,q')\in\transitionarrow$ holds, where $\alpha\in A^\tau$, and
we write $q\transrel{\alpha}$ as a short hand for $\exists q'\in Q : q\transrel{\alpha} q'$ and
$q\not\transrel{\alpha}$, else.
Furthermore, we write $q\xrightarrow{\alpha_1\cdots\alpha_n}q'$ to express that
$\exists q_0,\ldots,q_n\in Q: q=q_0\transrel{\alpha_1}q_1\transrel{\alpha_2}\cdots\transrel{\alpha_n}q_n=q'$ holds,
and write $q\Transrel{\epsilon}q'$ whenever $q=q'$ or $q\xrightarrow{\tau\cdots\tau}q'$.
Additionally, by $q\Transrel{\alpha}q'$, we denote that
$\exists q_1,q_2:q\Transrel{\epsilon}q_1\transrel{\alpha}q_2\Transrel{\epsilon}q'$.
We further use the notations $q\xRightarrow{a_1\cdots a_n}q'$ and $q\Transrel{a}$ ($a, a_1, \ldots, a_n \in A^{*}$)
analogously to $q\xrightarrow{\alpha_1\cdots\alpha_n}q'$ and $q\transrel{\alpha}$.
Finally, by $q_0\transrel{a_1}q_1\transrel{a_2}\cdots\transrel{a_n}q_n$
we denote a \emph{path}, where $\sigma=a_1a_2\ldots a_n\in A^*$
is called a \emph{trace} (note: $\tau$ equals $\epsilon$).
We identify an IOLTS with its initial state (\ie $q\in Q$ is the initial state of $q=(Q,I,O,\transitionarrow)$).
We only consider \emph{strongly convergent} IOLTS (\ie no infinite $\tau$-sequences exist).

In the \ioco testing theory, both specification $s$
as well as a (black-box) implementation under test $i$ are assumed to be
(explicitly or implicitly) given as IOLTS.
In particular, \ioco does not necessarily require specification $s$
to be \emph{input-enabled}, whereas implementation $i$ is
assumed to never reject any input $a\in I$ from the environment (or tester).
More precisely, \ioco requires implementations to be \emph{weak input-enabled}
(\ie $\forall q\in Q: \forall a\in I:q\Transrel{a}$)
thus yielding the subclass of \emph{I/O transition systems} (IOTS).
Intuitively, the IOTS of implementation $i$ \emph{I/O-conforms}
to the IOLTS of specification $s$ if all output behaviors of $i$ observed
after any possible sequence $\sigma = \alpha_1\cdots\alpha_n$
in $s$ are permitted by $s$.
In case of non-determinism, more than one
state may be reachable in $i$ as well as in $s$ after sequence $\sigma$ and therefore
all possible outputs of any state in the set
$$p\after\sigma:=\{q\in Q\mid p\Transrel{\sigma}q\}$$
have to be taken into account.
Formally, set $\out(Q')\subseteq O$ denotes all output actions being
enabled in any possible state $q\in Q' = p\after\sigma$.
To further reject trivial implementations never showing any outputs,
the notion of \emph{quiescence} has been introduced by means of a special
observable action $\delta$ explicitly denoting the permission of the \emph{absence}
(suspension) of any output in a state $p$, thus requiring an input to proceed.
In particular, $p$ is \emph{quiescent}, denoted $\delta(p)$,
iff 
$$\init(p):=\{\alpha\in(I\cup O\cup\{\tau\})\mid p\transrel{\alpha}\}\subseteq I$$ 
holds.
Thereupon, we denote
$$\out(P):=\{\alpha\in O\mid \exists p\in P:p\transrel{\alpha}\}\cup\{\delta\mid \exists p\in P:\delta(p)\},$$
where symbol $\delta$ is used both as action
as well as a state predicate.
Based on these notions, I/O conformance
is defined with respect to the set of \emph{suspension traces}
$$\straces(s):=\{\sigma\in(I\cup O\cup\{\delta\})^*\mid p\Transrel{\sigma}\}$$
of specification $s$, where $q\transrel{\delta}q$ iff $\delta(q)$.

\begin{definition}[\textbf{ioco}~\cite{Tretmans1996}]\label{def:ioco-init}
	Let $s$ be an IOLTS and $i$ an IOTS with identical sets $I$ and $O$.
	$$i\mathioco s:\Leftrightarrow\forall\sigma\in \straces(s):\out(i\after\sigma)\subseteq \out(s\after\sigma).$$
\end{definition}
%
\section{Modal Input/Output Conformance with Input Refusals}\label{sec:mioco}
IOLTS permit specifications $s$ to be \emph{underspecified} by means of
\emph{unspecified} input behaviors and
\emph{non-deterministic} input/output behaviors.
In particular, if $q\not\transrel{a}$, then
no proper reaction on occurrences of input
$a\in I$ is specified while residing in state $q$.
Moreover, if $q\transrel{a}q'$ and $q\transrel{a}q''$, $a\in \textit{A}^{\tau}$,
it does not necessarily follow that $q'=q''$ and if $q\transrel{a'}q'$ and
$q\transrel{a''}q''$ with $a',a''\in O$, it does not necessarily follow that $a'=a''$
(\ie IOLTS are neither input-, nor output-deterministic).
In this way, \textbf{ioco} permits, at least up to a certain degree,
implementation freedom in two ways.
First, in case of input behaviors being unspecified in $s$, \textbf{ioco}
solely relies on \emph{positive} testing principles, \ie
reactions to unspecified input behaviors
are never tested and may therefore
show \emph{arbitrary} output behaviors if ever applied to $i$.
Second, in case of non-deterministic specifications, implementation
$i$ is allowed to show \emph{any, but at least one} of those output behaviors
being permitted by $s$ (if any), or it must be quiescent, else.
These limitations of \textbf{ioco} in handling underspecified
behaviors essentially stem from the limited expressive power
of IOLTS.
To overcome these limitations, we propose to adopt richer
specification concepts from interfaces theories~\cite{Raclet2011}
to serve as novel formal foundation for I/O conformance testing.
In particular, we replace IOLTS by
a modified version of (I/O-labeled) Modal Interface Automata (MIA)
with universal state~\cite{Bujtor2015a}.
Similar to IOLTS, MIA also support both kinds of underspecification but allow
for explicit distinctions
(1) between obligatory and allowed behaviors in case of
	non-deterministic input/output behaviors, and
(2) between critical and non-critical unspecified input behaviors.

Concerning (1), MIA separate mandatory from optional
behaviors in terms of may/must transition modality.
For every \emph{must}-transition $q\oversetmust{a}q'$,
a corresponding \emph{may}-transition $q\oversetmay{a}q'$ exists, as
mandatory behaviors must also be allowed (so-called syntactic consistency).
Conversely, may-transitions $q\oversetmay{a}q'$ for which
$q\not\oversetmust{a}q'$ holds constitute
optional behaviors.
Accordingly, we call may-transitions without corresponding
must-transitions \emph{optional}, else \emph{mandatory}.

Concerning (2), MIA make explicit input actions
$a\in I$ being unspecified, yet uncritical
in a certain state $q$ by introducing may-transitions $q\oversetmay{a}u$
leading to a special \emph{universal state} $u$ (permitting any possible
behavior following that input).
In contrast, unintended input actions to be rejected
in a certain state are implicitly forbidden
if $q\not\oversetmay{a}$ holds.
We alter the interpretation of
unspecified input behaviors of MIA by introducing
a distinct \emph{failure state} $q_\Phi$ replacing $u$.
As a consequence, an unspecified input $a\in I$ being uncritical
if residing in a certain state $q$ is (similar to IOLTS) implicitly
denoted as $q\not\oversetmay{a}$,
whereas inputs $a'\in I$ being critical while
residing in state $q$ are explicitly forbidden by $q\oversetmust{a'} q_\Phi$.
We therefore enrich I/O conformance testing by
the notion of \emph{input refusals} in the spirit of
refusal testing, initially proposed by Phillips
for testing preorders on LTS with undirected actions~\cite{Phillips1987}.
Analogous to quiescence, denoting the \emph{observable absence of any output}
in a certain state, refusals therefore denote
the \emph{observable rejection of a particular input} in a certain state
during testing.
In this way, we unify positive testing
(\ie unspecified behaviors are ignored)
and negative testing
(\ie unspecified behaviors must be rejected)
with optimistic and pessimistic environmental
assumptions known from interface theories~\cite{Raclet2011}.
In particular, we are now able to explicitly reject certain input behavior, which is not supported by \ioco.
We refer to the resulting model as
\emph{Modal Interface Automata with Input Refusals (IR-MIA)}.

\begin{definition}[IR-MIA]\label{def:ir-mia}
A \emph{Modal Interface Automaton with Input-Refusal} (IR-MIA or \miaforbidden)
is a tuple $(Q,I_Q,O_Q,\mustarrow,\mayarrow,q_\Phi)$, where $Q$ is
a finite set of \emph{states} with \emph{failure state} $q_\Phi\in Q$,
$A_Q=I_Q\cup O_Q$ is a finite set of \emph{actions} with
		$\tau\notin A_Q$ and $I_Q\cap O_Q = \emptyset$ and for all $a\in A_Q\cup\{\tau\}, i\in I_Q$,
	\begin{enumerate}
		\item $\mustarrow\subseteq ((Q\setminus\{q_\Phi\}) \times I_Q \times Q) \cup ((Q\setminus\{q_\Phi\}) \times (O_Q\cup\{\tau\}) \times (Q\setminus\{q_\Phi\}))$,
		\item $\mayarrow\subseteq ((Q\setminus\{q_\Phi\}) \times I_Q \times Q) \cup ((Q\setminus\{q_\Phi\}) \times (O_Q\cup\{\tau\}) \times (Q\setminus\{q_\Phi\}))$,
		\item $q\oversetmust{a}q' \Rightarrow q\oversetmay{a}q'$,
		\item $q\oversetmay{i}q_\Phi \Leftrightarrow q\oversetmust{i}q_\Phi$, and
		\item $q\oversetmust{i}q_\Phi \Rightarrow \left(\forall q'.q\oversetmay{i}q' \Rightarrow q'= q_\Phi \right)$.
	\end{enumerate}
\end{definition}

Property~3 ensures syntactic consistency and
properties~1 and~2 together with property~4 ensure that
the failure state $q_\Phi$ only occurs as target of must-transitions being
labeled with input actions.
Property~5 further requires consistency of refusals
of specified input actions in every state (\ie each input is either forbidden or not,
but not both in a state $q$).

Figure~\ref{fig:ir-mia-impl} shows a sample IR-MIA.
Dashed lines denote optional behaviors and solid lines denote mandatory behaviors.
Additionally, the distinct state $q_\Phi$ depicts the failure state
(\ie input \emph{f} is refused by state $q_1$).
This example also exhibits input non-determinism
(state $q_0$ defines two possible reactions to input \emph{a}),
as well as output non-determinism
(state $q_3$ defines two possible outputs after input \emph{d}).

\begin{figure}[tp]
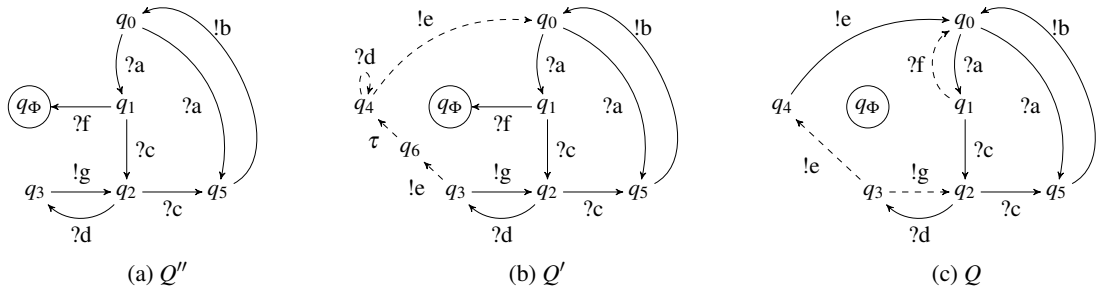

	\hfill
	\subfloat[$Q''$]{\label{fig:ir-mia-impl-2}\input{figures/ir-mia-impl-2.tex}}
	\hfill
	\subfloat[$Q'$]{\label{fig:ir-mia-impl}\input{figures/ir-mia-impl.tex}}
	\hfill
	\subfloat[$Q$]{\label{fig:ir-mia-spec}\input{figures/ir-mia-spec.tex}}
	\hfill\strut
	\caption{Sample IR-MIA}\label{fig:ir-mia}
\end{figure}

\emph{Modal refinement} provides a
semantic implementation relation on MIA~\cite{Bujtor2015a}.
Intuitively, MIA $P$ refines MIA $Q$ if mandatory behaviors of $Q$
are preserved in $P$ and optional behaviors in $P$ are permitted
by $Q$.
Adapted to IR-MIA, input behaviors being unspecified
in $Q$, may be either implemented arbitrarily in $P$, or become forbidden
after refinement.
In particular, if $q\not\oversetmay{a}$ holds in $Q$, then either
$q\not\oversetmay{a}$, $q\oversetmay{a}q'$ (and even $q\oversetmust{a}q'$), or
$q\oversetmust{a}q_{\Phi}$ holds in $P$, respectively.

\begin{definition}[IR-MIA Refinement]\label{def:mia-refinement-forbidden}
  	Let $P, Q$ be \miaforbidden with $I_P=I_Q$ and $O_P=O_Q$.
  	A relation $\mathcal{R}\subseteq P\times Q$ is an
		\emph{IR-MIA Refinement Relation}
		if for all $(p,q)\in\mathcal{R}$ and $\omega\in(O\cup\{\tau\})$, with $p\neq p_\Phi$ and $\gamma\in\{\Diamond,\Box\}$, it holds that
  	\begin{enumerate}
  		\item\label{refine-failure} $q\neq q_\Phi$,
  		\newcounter{mia-ref-must-i}
		\setcounter{mia-ref-must-i}{\value{enumi}}
  		\item\label{refine-must-i} $q\oversetmust{i}q'\neq q_\Phi$ implies $\exists p'.p\oversetmust{i}\Oversetmust{\epsilon}p'\neq p_\Phi$ and $(p',q')\in\mathcal{R}$,
  		\newcounter{mia-ref-must-o}
		\setcounter{mia-ref-must-o}{\value{enumi}}
  		\item\label{refine-must-o} $q\oversetmust{\omega}q'$ implies $\exists p'.p\Oversetmust{\hat\omega}p'$ and $(p',q')\in\mathcal{R}$,
  		\item\label{refine-may-i-consistency} $p\oversetmay{i}p' \wedge q\oversetmay{i}$ implies $\exists q'.q \oversetmay{i}\Oversetmay{\epsilon}q'$ and $(p',q')\in\mathcal{R}$,
  		\item\label{refine-may-i} $q\oversetmay{i}q'$ implies $\exists p'.p\oversetmay{i}\Oversetmay{\epsilon}p'$ and $(p',q')\in\mathcal{R}$, and
  		\item\label{refine-may-o} $p\oversetmay{\omega}p'$ implies $\exists q'.q\Oversetmay{\hat\omega}q'$ and $(p',q')\in\mathcal{R}$.
  	\end{enumerate}
  	State $p$ refines state $q$ if there exists $\mathcal{R}$ such that $(p,q)\in\mathcal{R}$.
  	(Note: $q\Transrel{\hat\omega}_\gamma q'$ equals $q\Transrel{o}_\gamma q'$ for $\hat\omega = o \in O$ and $q\Transrel{\epsilon} q'$ otherwise).
\end{definition}

Clause 1 ensures that the failure state $q_\Phi$ can only be refined by $p_\Phi$, since both suspend any subsequent behavior.
Clauses 2 and 3 guarantee that mandatory behavior of $Q$ is preserved by $P$.
All other clauses handle optional behavior, where inputs are either refined to
forbidden or mandatory inputs, and outputs are either refined to mandatory or unspecified outputs.
By $P \miarefforbidden Q$, we denote the existence of an IR-MIA refinement relation between $P$ and $Q$.

As an example, consider IR-MIA $Q$ and $Q'$ in Fig.~\ref{fig:ir-mia}.
$Q'\miarefforbidden Q$ does not hold as the mandatory
output \emph{e} of $q_4$ in $Q$ is not mandatory anymore in $Q'$.
However, the other modifications in $Q'$ are valid refinements of $Q$
as output \emph{g} of $q_3$
has become mandatory, and optional
input \emph{f} of $q_1$ is now refused (\ie
the transition is redirected to $q_\Phi$).
Additionally, inputs being unspecified in $Q$ may be added to $Q'$
(\eg $q_4$ of $Q'$ now accepts input \emph{d}).
Furthermore, internal steps may be added after inputs as well as before
and after outputs under refinement (\eg $Q'$ has a $\tau$ step
after output \emph{e} in $q_3$).
The former ensures that a refined IR-MIA may be controlled by the environment in the same way as
the unrefined IR-MIA.
Considering IR-MIA $Q''$ in Fig.~\ref{fig:ir-mia} instead, $Q''\miarefforbidden Q$ holds.
The removal of mandatory output \emph{e} from $q_4$ is valid as $q_4$ is
not reachable anymore after refinement.

In the context of modal I/O conformance testing, modal
refinement offers a controlled way to resolve
underspecification within specifications $s$.
In addition, we also assume $i$ to be represented as IR-MIA
in order to support (partially) underspecified implementations under test
as apparent in earlier phases of continuous systems and component development.

We next define an adapted version of \ioco to operate on IR-MIA.
Intuitively, a modal implementation $i$ \emph{I/O-conforms}
to a modal specification $s$ if all
observable mandatory behaviors of $s$ are also observable
as mandatory behaviors of $i$ and none of the observable optional behaviors
of $i$ exceed the observable optional behaviors of $s$.
If established between implementation $i$ and specification $s$,
modal I/O conformance ensures for all implementations $i'\miarefforbidden i$, derivable from
$i$ via modal refinement, the existence of an accompanying specification refinement $s'\miarefforbidden s$ of $s$
such that $i'$ is I/O conforming to $s'$.

Similar to $\delta$ denoting observable quiescence, we introduce
a state predicate $\varphi$ to denote
\emph{may-failure/\linebreak must-failure} states (\ie states having may/must input-transitions
leading to $q_\Phi$).
We therefore use $\varphi$ as a special symbol to
observe refusals of particular inputs in certain states of the implementation during testing.
To this end, we first lift the auxiliary notations of \textbf{ioco} from IOLTS
to IR-MIA, where we write $\gamma\in\{\Diamond,\Box\}$ for short in the following.

\begin{definition}\label{def:modal-init-forbidden}
	Let $Q$ be a \miaforbidden over $I$ and $O$, $p\in Q$ and $\sigma\in(I\cup O \cup \{ \delta,\varphi\})^*$.
	\begin{itemize}
		\item $\init_\gamma(p):=\{\mu\in(I\cup O)\mid p\transrel{\mu}_\gamma\} \cup \{\varphi\mid p=p_\Phi\}$,
		\item $p$ is \emph{may-quiescent}, denoted by $\delta_\Diamond(p)$, iff $\mustinit(p)\subseteq I$, $p\not\oversetmust{\tau}$, and $p\neq p_\Phi$,
		\item $p$ is \emph{must-quiescent}, denoted by $\delta_\Box(p)$, iff $\mayinit(p)\subseteq I$, $p\not\oversetmay{\tau}$, and $p\neq p_\Phi$,
		\item $p$ is \emph{may-failure}, denoted by $\varphi_\Diamond(p)$, iff $p=p_\Phi$ or $\exists p'\in Q:(p''\oversetmay{i}p\land p''\not\oversetmust{i}p)$,
		\item $p$ is \emph{must-failure}, denoted by $\varphi_\Box(p)$, iff $p=p_\Phi$,
		\item $p\after_\gamma\sigma:=\{p' \mid p\Transrel{\sigma}_\gamma p'\}$,
		\item $\textit{Out}_\gamma(p):=\{\mu\in O\mid p\transrel{\mu}_\gamma\}\cup\{\delta\mid\delta_\gamma(p)\}\cup\{\varphi\mid\varphi_\gamma(p)\}$, and
		\item $\textit{Straces}_\gamma(p):=\{\sigma\in(I\cup O\cup\{\delta,\varphi\})^*\mid p\Transrel{\sigma}_\gamma\}$, where $p\transrel{\delta}_\gamma p$ if $\delta_\gamma(p)$, and $p\transrel{\varphi}_\gamma p$ if $\varphi_\gamma(p)$.
	\end{itemize}
\end{definition}

Hence, quiescence as well as failure behaviors may occur with both may- and must-modality.
Intuitively, a state is may-quiescent if all enabled output transitions are optional, \ie such a state {\em may} become quiescent under refinement.
Likewise, a state $p$ is a may-failure if there is an optional input leading to $p$, since this optional input may be refused under refinement.

According to \textbf{ioco}, $\miaforbidden$ $i$
constituting a modal implementation under test
is assumed to be input-enabled.
In particular, \emph{modal input-enabledness} of
IR-MIA comes in four flavors by combining weak/\linebreak strong
input-enabledness with may-/must-modality.
Note, that $q\transrel{i}_\gamma q'$ implies $q\Transrel{i}_\gamma q'$, $q\Oversetmust{i}q'$
implies $q\Oversetmay{i}q'$, and $q\oversetmust{i}q'$ implies $q\oversetmay{i}q'$.

\begin{definition}[Input-Enabled IR-MIA]\label{def:mia-input-enabledness}
	$\miaforbidden$ $Q$ is \emph{weak/strong $\gamma$-input-enabled}, respectively, iff for each $q\in Q\setminus\{q_\Phi\}$
	it holds that $\forall i\in I:\exists q'\in Q:q\Transrel{i}_\gamma q'$, or
	$\forall i\in I:\exists q'\in Q:q\transrel{i}_\gamma q'$.
\end{definition}

May-input-enabledness is preserved under modal refinement
as optional input behaviors either remain optional,
become mandatory, or are redirected to the failure state (and
finally become must-input-enabled under complete refinement).

\newcounter{may-input-enabledness-preservation-counter}
\setcounter{may-input-enabledness-preservation-counter}{\value{lemma}}
\begin{lemma}\label{lemma:may-input-enabledness-preservation}
If \miaforbidden $i$ is strong may-input-enabled
then $i'\miarefforbidden i$ is strong may-input-enabled.
\end{lemma}

We now define a modal version of \textbf{ioco}
on IR-MIA (called \textbf{modal-irioco} or $\mathmiocophi$), by means of
alternating suspension-trace inclusion.

\begin{definition}[\textbf{modal-irioco}]\label{def:mioco}
	Let $s$ and $i$ be \miaforbidden over $I$ and $O$ with $i$ being weak may-input-enabled.
	$i\mathmiocophi s:\Leftrightarrow$
	\begin{enumerate}
		\item $\forall\sigma\in\maystraces(s):\mayout(i\mayafter\sigma)\subseteq\mayout(s\mayafter\sigma)$, and
		\item $\forall\sigma\in\maystraces(i):\mustout(s\mayafter\sigma)\subseteq\mustout(i\mayafter\sigma)$.
	\end{enumerate}
\end{definition}

We illustrate the intuition of \textbf{modal-irioco} by providing a concrete example.
Let IR-MIA in Fig.~\ref{fig:ir-mia-impl} constitute
implementation $i$ and IR-MIA in Fig.~\ref{fig:ir-mia-spec} constitute specification $s$.
Similar to \textbf{ioco},
Property~1 of \textbf{modal-irioco} requires all possible output behaviors of $i$ to be
permitted by $s$ which is satisfied in this example.
Property~2 of \textbf{modal-irioco} requires all mandatory outputs of $s$ to be
actually implemented as mandatory outputs in $i$.
This property does not hold in the example as mandatory output \emph{e}
of $q_4$ in $s$ is not mandatory in $i$.
As a consequence, $i\mathmiocophi s$ does not hold.
The example in Fig.~\ref{fig:ir-mia} also explains
why we consider \maystraces and \mayafter in property~2
(unlike \textbf{modal-ioco} in~\cite{Lochau2014}).
Otherwise, output \emph{e} of $q_4$ in $i$
would not be considered as mandatory output behavior
because $q_4$ is not reachable via must-transitions.
In contrast, when considering the IR-MIA in Fig.~\ref{fig:ir-mia-impl-2} as $i$
and the IR-MIA in Fig.~\ref{fig:ir-mia-spec} as $s$, we have
$i\mathmiocophi s$ as the
mandatory output \emph{e} of $q_4$ in $s$ is not reachable in $i$.

Figure~\ref{fig:universal-example-i}, \ref{fig:universal-example-u} and \ref{fig:universal-example-phi}
illustrate the necessity for re-interpreting universal state $u$ of MIA~\cite{Bujtor2015a} as failure state
$q_\Phi$ in IR-MIA.
The IR-MIA in Fig.~\ref{fig:universal-example-i} serves as implementation $i$, the MIA in
Fig.~\ref{fig:universal-example-u} serves as specification $s_u$ with universal state,
and the IR-MIA in Fig.~\ref{fig:universal-example-phi} depicts the same specification with failure state
$s_\Phi$ instead of $u$.
Hence, $i$ would be (erroneously) considered to be non-conforming to $s_u$ as
state $u$ does not specify any outputs (\ie $u$ is quiescent).
In contrast, we have $i\mathmiocophi s_\Phi$ as the reaction of $i$
to input \emph{a} is never tested, because this input is unspecified in $s_\Phi$.

\begin{figure}[tp]
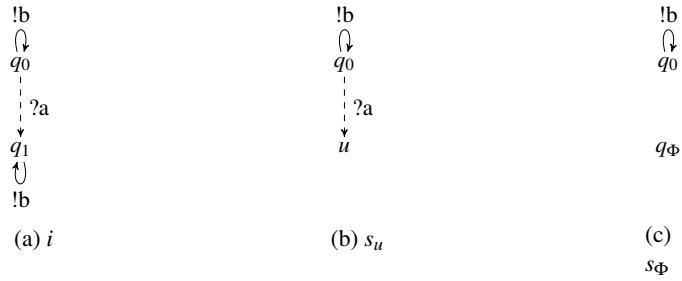

	\hfill
	\subfloat[$i$]{\label{fig:universal-example-i}\input{figures/universal-example-i.tex}}
	\hfill
	\subfloat[$s_u$]{\label{fig:universal-example-u}\input{figures/universal-example-u.tex}}
	\hfill
	\subfloat[$s_\Phi$]{\label{fig:universal-example-phi}\input{figures/universal-example-phi.tex}}
	\hfill\strut
	\caption{Problem of modal-irioco regarding MIA with universal state.}\label{fig:universal-example}
\end{figure}

An I/O conformance testing theory
is \emph{correct} if it is \emph{sound}
(\ie every implementation $i$ conforming
to specification $s$ does indeed only show specified behaviors),
and \emph{complete} (\ie every erroneous implementation $i$ is rejected)~\cite{Tretmans1996}.
For lifting these notions to IR-MIA,
we relate \textbf{modal-irioco} to \textbf{ioco}.
This way, we show compatibility of \textbf{modal-irioco} and the original \textbf{ioco} as follows.
\begin{itemize}
	\item \textbf{modal-irioco} is \emph{sound} if $i\mathmiocophi s$ implies that every refinement
of $i$ conforms to a refinement of $s$ with respect to \ioco.
  \item \textbf{modal-irioco} is \emph{complete} if the correctness of all refinements of $i$
regarding $s$ with respect to \ioco implies $i\mathmiocophi s$, and if at least
one refinement of $i$ is non-conforming to any refinement of $s$,
then $i\mathmiocophi s$ does not hold.
\end{itemize}

We first have to show that \textbf{modal-irioco} is preserved
under modal refinement.
Although, intuitions behind both relations
are quite similar, they are incomparable.
Figure~\ref{fig:mioco-miaref-preservation} gives examples showing that both $i\mathmiocophi s\Rightarrow i\miarefforbidden s$ and $i\miarefforbidden s\Rightarrow i\mathmiocophi s$ do not hold.
Firstly, we take a look at Figs.~\ref{fig:mioco-miaref-preservation-i1} and~\ref{fig:mioco-miaref-preservation-s1}.
Here $i_1\mathmiocophi s_1$ holds.
Note, that $\delta\in\mayout(s_1\mayafter !o\cdot?i)$ because the state on the right-hand side after the input \emph{i} is may-quiescent.
Though, $i_1\miarefforbidden s_1$ does not hold $i_1$ has the output \emph{o'} on the left-hand side, and it does not have \emph{o'} on the right side.
Therefore, $i\mathmiocophi s\nRightarrow i\miarefforbidden s$.
Secondly, consider Figs.~\ref{fig:mioco-miaref-preservation-i2} and~\ref{fig:mioco-miaref-preservation-s2}.
In this case, $i_2\miarefforbidden s_2$ holds because in $s_2$ the input \emph{i} on the right-hand side is underspecified and may be implemented arbitrarily.
Here, $i_2\mathmiocophi s_2$ does not hold because $\mayout(i_2\mayafter !o\cdot?i)\nsubseteq\mayout(s_2\mayafter !o\cdot?i)$.
Therefore, $i\miarefforbidden s\Rightarrow i\mathmiocophi s$ is also not true.
To conclude \miaforbidden refinement and \mioco are incomparable.
\begin{figure}[tp]
	\hfill
	\subfloat[Implementation 1]{\label{fig:mioco-miaref-preservation-i1}\input{figures/mioco-miaref-preservation-i1.tex}}
	\hfill
	\subfloat[Specification 1]{\label{fig:mioco-miaref-preservation-s1}\input{figures/mioco-miaref-preservation-s1.tex}}
	\hfill
	\subfloat[Implementation 2]{\label{fig:mioco-miaref-preservation-i2}\input{figures/mioco-miaref-preservation-i2.tex}}
	\hfill
	\subfloat[Specification 2]{\label{fig:mioco-miaref-preservation-s2}\input{figures/mioco-miaref-preservation-s2.tex}}
	\hfill\strut
	\caption{Figures~\ref{fig:mioco-miaref-preservation-i1} and~\ref{fig:mioco-miaref-preservation-s1} are a counterexample to $i\mathmiocophi s\Rightarrow i\miarefforbidden s$ because $i_1\mathmiocophi s_1$ but not $i_1\miarefforbidden s_1$.
	Figures~\ref{fig:mioco-miaref-preservation-i2} and~\ref{fig:mioco-miaref-preservation-s2} are a counterexample to $i\miarefforbidden s\Rightarrow i\mathmiocophi s$ because $i_2\miarefforbidden s_2$ but not $i_2\mathmiocophi s_2$.}\label{fig:mioco-miaref-preservation}
\end{figure}
Instead, we obtain a weaker correspondence.

\newcounter{mioco-refinement-compatability-counter}
\setcounter{mioco-refinement-compatability-counter}{\value{theorem}}
\begin{theorem}\label{theorem:mioco-refinement-compatability}
	Let $i,s$ be \miaforbidden, $i$ being weak may-input-enabled
	and $i\mathmiocophi s$.
	Then for each $i'\miarefforbidden i$ there exists
	$s'\miarefforbidden s$ such that $i'\mathmiocophi s'$ holds.
\end{theorem}

Note, that we refer to the \textbf{ioco}-relation in our \textbf{modal-irioco} in Clause 1 in Def.~\ref{def:mioco}.
Hence, in order to relate \textbf{modal-irioco} and \textbf{ioco}, we
define applications of \textbf{ioco} to IR-MIA by
considering the may-transition relation as the actual transition relation.

\begin{definition}[\ioco on \miaforbidden]\label{def:ioco-mia}
	Let $i$, $s$ be \miaforbidden, $i$ be weak may-input-enabled.
	Then, $i\mathioco s:\Leftrightarrow\forall\sigma\in \maystraces(s):\mayout(i\mayafter\sigma)\subseteq \mayout(s\mayafter\sigma)$.
\end{definition}

Based on this definition, we are able to prove correctness of \textbf{modal-irioco}.

\newcounter{soundness-completeness-mioco-counter}
\setcounter{soundness-completeness-mioco-counter}{\value{theorem}}
\begin{theorem}[\textbf{modal-irioco} is correct]\label{theorem:soundness-completeness-mioco}
	Let $i,s$ be \miaforbidden, $i$ being weak may-input-enabled.
	\begin{enumerate}
		\item If $i\mathmiocophi s$, then for all $i'\miarefforbidden i$,
		there exists $s'\miarefforbidden s$ such that $i'\mathioco s'$.
		\item If there exists $i'\miarefforbidden i$
		such that $i'\mathioco s'$ does not hold for any $s'\miarefforbidden s$,
		then $i\mathmiocophi s$ does not hold.
	\end{enumerate}
\end{theorem}

Property~1 states soundness of \textbf{modal-irioco}.
However, the immediate inverse does not hold
as \textbf{ioco} does not guarantee mandatory
behaviors of $s$ to be actually implemented by $i$
(cf. Fig.~\ref{fig:soundness-converse-counter} for a counter-example where $i\mathioco s$ but not $i\mathmiocophi s$).
Instead, Property 2 states completeness of \textbf{modal-irioco}
in the sense that modal implementations $i$ are rejected if
at least one refinement of $i$ exists not conforming
to any refinement of specification $s$.
Finally, we conclude that $\miocophi$
becomes a preorder if being restricted
to input-enabled IR-MIA specifications.

\begin{figure}[tp]
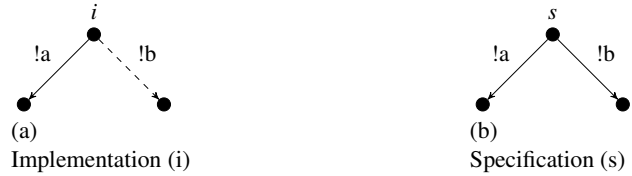

	\hfill
	\subfloat[\mbox{Implementation (i)}]{\label{fig:soundness-converse-counter-i}\input{figures/soundness-converse-counter-i.tex}}
	\hfill
	\subfloat[\mbox{Specification (s)}]{\label{fig:soundness-converse-counter-s}\input{figures/soundness-converse-counter-s.tex}}
	\hfill\strut
	\caption{Example for which the converse of soundness does not hold.
	For each variant $i'\miarefforbidden i$ it holds that $i'\mathioco s$ but $i\mathmioco s$ does not hold.
	Note, that $I=\emptyset$ such that the implementation is must-input-enabled.}\label{fig:soundness-converse-counter}
\end{figure}

\newcounter{mioco-preorder-counter}
\setcounter{mioco-preorder-counter}{\value{theorem}}
\begin{theorem}\label{theorem:mioco-preorder}
	\miocophi is a preorder on the set of weak may-input-enabled \miaforbidden.
\end{theorem}

Must-input-enabledness (and therefore may-input-enabledness) of a specification $s$
may be achieved for any given IR-MIA by applying a
behavior-preserving canonical \emph{input completion}, while still
allowing arbitrary refinements of
previously unspecified inputs (instead of ignoring inputs as, \eg achieved by \emph{angelic completion}~\cite{Vaandrager1991}).
This construction essentially adapts the
notion of \emph{demonic completion}~\cite{Nicola1995}
from IOLTS to IR-MIA as follows.

\begin{figure}[tp]
	\hfill
	\subfloat[$s$]{\label{fig:demonic-completion-s}\input{figures/demonic-completion-s.tex}}
	\hfill
	\subfloat[$\Xi(s)$]{\label{fig:demonic-completion-sdc}\input{figures/demonic-completion-sdc.tex}}
	\hfill\strut
	\caption{Demonic completion adapted for \miaforbidden and \miocophi.}\label{fig:demonic-completion}
\end{figure}

\begin{definition}[Demonic Completion of IR-MIA]\label{def:demonic-completion}
	The \emph{demonic completion} of \miaforbidden $(Q,I,O,\mustarrow,\mayarrow,q_\Phi)$ with $\forall q\in Q: q\oversetmay{\tau}\Rightarrow q\oversetmust{\tau}$
	is a \miaforbidden $(Q',I,O,\mustarrow',\mayarrow',q_\Phi)$, where
	\begin{itemize}
		\item $Q'=Q\cup\{q_\chi,q_\Omega\}$ with $q_\chi,q_\Omega\notin Q$, and
		\item $\mustarrow'=\mustarrow\cup\{(q,i,q_\chi)\mid q\in Q, i\in I, q\not\oversetmust{i},q\not\oversetmust{\tau}\}\cup\{(q_\chi,\tau,q_\Omega)\}\cup\{(q_\chi,\lambda,q_\chi),(q_\Omega,\lambda,q_\Omega)\mid \lambda\in I\}$.
		\item $\mayarrow'=\mayarrow\cup\{(q,i,q_\chi)\mid q\in Q, i\in I, q\not\oversetmust{i},q\not\oversetmust{\tau}\}\cup\{(q_\chi,\tau,q_\Omega)\}\cup\{(q_\Omega,\lambda,q_\chi)\mid \lambda\in(I\cup O)\}\cup\{(q_\chi,\lambda,q_\chi),(q_\Omega,\lambda,q_\Omega)\mid \lambda\in I\}$.
	\end{itemize}
\end{definition}

The restriction imposed by $\forall q\in Q: q\oversetmay{\tau}\Rightarrow q\oversetmust{\tau}$
is due to weak input-enabled states not being input-enabled anymore if an optional $\tau$-transition is removed.
We refer to the demonic completion of \miaforbidden $s$ as $\Xi(s)$.

Figure~\ref{fig:demonic-completion} illustrates demonic completion.
As state $q_1$ of $s$ is not must-input-enabled,
a must-transition for action \textit{a} is added from $q_1$ to $q_\chi$.
The fresh states $q_\chi$ and $q_\Omega$ have outgoing must-transitions
for each $i\in I$, thus being (strong) must-input-enabled.
Additionally, $q_\chi$ in combination with $q_\Omega$ allow (but do not require)
every output $o\in O$ (in $q_\chi$ via one silent move), such that demonic
completion preserves underspecification.
We conclude that this construction preserves \textbf{modal-irioco}.

\newcounter{mioco-demonic-completion-counter}
\setcounter{mioco-demonic-completion-counter}{\value{theorem}}
\begin{theorem}\label{theorem:mioco-demonic-completion}
	Let $i$, $s$ be \miaforbidden with $i$ being weak must-input-enabled.
	Then $i\mathmiocophi \Xi(s)$ if $i\mathmiocophi s$.
\end{theorem}

%
\section{Compositionality}\label{sec:compositionality}

Interface theories are equipped
with a (binary) \emph{interleaving parallel operator} $\parallel$
on interface specifications to define
interaction behaviors in systems
composed of multiple concurrently running components~\cite{Alfaro2001}.
Intuitively, transition $p\transrel{a}p'$, $a\in O_{P}$, of component $P$ synchronizes with
transition $q\transrel{a}q'$, $a\in I_{Q}$, of component $Q$, where the resulting synchronized action
$(p,q)\transrel{\tau}(p',q')$ becomes a silent move.
Modal interface theories generalize parallel composition
to \emph{multicast} communication (\ie one output action
synchronizes with all concurrently running
components having this action as input)
and explicit \emph{hiding} of synchronized
output actions~\cite{Raclet2011}.
According to MIA, we define parallel composition
on IR-MIA in two steps:
(1) standard parallel product $P_1 \otimesforbidden P_2$
on $\miaforbidden$ $P_1$, $P_2$, followed by
(2) parallel composition $P_1\parallelforbidden P_2$,
removing erroneous states $(p_1,p_2)$ from
$P_1 \otimesforbidden P_2$, where for an output
action of $p_1$, no corresponding input
is provided by $p_2$ (and vice versa).
In addition, all states $(p_1',p_2')$ from which
erroneous states are reachable are also removed (pruned)
from $P_1\parallelforbidden P_2$.

Concerning (1), we first require \emph{composability}
of $P_1$ and $P_2$ (\ie disjoint output actions).
In $P_1 \otimesforbidden P_2$, a fresh state $p_{12\Phi}$
serves as unified failure state.
The input alphabet of $P_1 \otimesforbidden P_2$ contains
all those inputs of $P_1$ and $P_2$ not being contained
in one of their output sets, whereas the output alphabet
of $P_1 \otimesforbidden P_2$ is the union
of both output sets.
The modality $\gamma$ of composed transitions
$(p_1,p_2)\transrel{\alpha}_\gamma (p_1',p_2')$
depends on the modality of the individual transitions.

\begin{definition}[IR-MIA Parallel Product]\label{def:parallel-product}
	\miaforbidden $P_1$, $P_2$ are \emph{composable} if $O_1 \cap O_2 = \emptyset$.
	The \emph{parallel product} is defined as $P_1\otimesforbidden P_2 = ((P_1\times P_2)\cup\{q_{12\Phi}\},I,O,\mustarrow,\mayarrow,p_{12\Phi})$, where $I=_\textit{def}(I_1\cup I_2)\setminus(O_1\cup O_2)$ and $O =_\textit{def} O_1\cup O_2$, and where $\mustarrow$ and $\mayarrow$ are the least relations satisfying the following conditions:
	\begin{tabbing}
		(May2/Must2)\quad \= $(p_1,p_2)\transrel{\alpha}_\gamma(p_1',p_2')$\quad \= if\quad \= \kill
		(May1/Must1) \> $(p_1,p_2)\transrel{\alpha}_\gamma (p_1',p_2)$ \> if \> $p_1\transrel{\alpha}_\gamma p_1'$ and $\alpha\notin A_2$ \\
		(May2/Must2) \> $(p_1,p_2)\transrel{\alpha}_\gamma (p_1,p_2')$ \> if \> $p_2\transrel{\alpha}_\gamma p_2'$ and $\alpha\notin A_1$ \\
		(May3/Must3) \> $(p_1,p_2)\transrel{a}_\gamma (p_1',p_2')$ \> if \> $p_1\transrel{a}_\gamma p_1'$ and $p_2\transrel{a}_\gamma p_2'$ for some $a$ \\
		(May4/Must4) \> $(p_1,p_2)\transrel{a}_\gamma p_{12\Phi}$ \> if \> $p_1\transrel{a}_\gamma p_1'$ and $p_2\not\transrel{a}_\gamma$ for some $a\in I_1\cap A_2$ \\
		(May5/Must5) \> $(p_1,p_2)\transrel{a}_\gamma p_{12\Phi}$ \> if \> $p_2\transrel{a}_\gamma p_2'$ and $p_1\not\transrel{a}_\gamma$ for some $a\in I_2\cap A_1$.
	\end{tabbing}
\end{definition}

Rules \emph{(May1/Must1)} and \emph{(May2/Must2)} define interleaving of 
transitions labeled with actions being exclusive to one of both components; 
whereas Rule \emph{(May3/Must3)} 
synchronizes transitions with common actions, and the Rules \emph{(May4/Must4)} and \emph{(May5/Must5)} 
forbid transitions of a component labeled with inputs being common to both components, but not being supported by the other component.
Concerning (2), we define $E\subseteq P_1\times P_2$
to contain \emph{illegal state pairs} $(p_1,p_2)$ in $P_1\otimesforbidden P_2$.

\begin{definition}[Illegal State Pairs]\label{def:illegal-states}
Given a parallel product $P_1\otimesforbidden P_2$, a state $(p_1,p_2)$ is a \emph{new error}
if there exists $a\in A_1\cap A_2$ such that
		\begin{itemize}
			\item $a\in O_1$, $p_1\oversetmay{a}$ and $p_2\not\oversetmust{a}$, or
			\item $a\in O_2$, $p_2\oversetmay{a}$ and $p_1\not\oversetmust{a}$, or
			\item $a\in O_1$, $p_1\oversetmay{a}$ and $p_2\oversetmay{a}p_{2\Phi}$, or
			\item $a\in O_2$, $p_2\oversetmay{a}$ and $p_1\oversetmay{a}p_{1\Phi}$.
			\smallskip
		\end{itemize}
	The relation $E\subseteq P_1\times P_2$ containing
	\emph{illegal state pairs} is the least relation such that $(p_1,p_2)\in E$ if
	\begin{itemize}
		\item $(p_1,p_2)$ is a new error, or
		\item $(p_1,p_2)\oversetmust{\omega}(p_1',p_2')$ with $\omega\in(O\cup\{\tau\})$ and $(p_1',p_2')\in E$.
	\end{itemize}
\end{definition}

If the initial state of $P_1\otimesforbidden P_2$ is illegal (\ie $(p_{01},p_{02})\in E$),
it is replaced by a fresh initial state without incoming and outgoing
transitions such that $P_1$ and $P_2$ are considered \emph{incompatible}.

\begin{definition}[IR-MIA Parallel Composition]\label{def:parallel-composition}
	The \emph{parallel composition} $P_1\parallelforbidden P_2$ of
	$P_1\otimesforbidden P_2$ is obtained by pruning illegal states as follows.
	\begin{itemize}
	\item transitions leading to a state of the form
	$(q_{1\Phi},p_2)$ or $(p_1,q_{2\Phi})$ are redirected to $q_{12\Phi}$.
	\item states $(p_1,p_2)\in E$  and all unreachable states
	(except for $q_{12\Phi}$) and all their incoming and
	outgoing transitions are removed.
	\item for states $(p_1,p_2)\notin E$
	and $(p_1,p_2)\oversetmay{i}(p_1',p_2')\in E$, $i\in I$,
	all transitions $(p_1,p_2)\oversetmay{i}(p_1'',p_2'')$ are removed.
	\end{itemize}
	If $(p_1,p_2)\in P_1\parallelforbidden P_2$, we write $p_1\parallelforbidden p_2$ and call $p_1$ and $p_2$ \emph{compatible}.
\end{definition}

For example, consider $P'=Q\parallelforbidden D$
(cf. Fig.~\ref{fig:ir-mia-spec} and Fig.~\ref{fig:comp-example}).
Here, $q_0$ of both $Q$ and $D$ have action \emph{a} as
common action thus being synchronized to become an output action in $P'$
(to allow multicast communication).
Action \emph{a} is mandatory in $P'$ as \emph{a} is mandatory in both $Q$ and $D$.
In any other case, the resulting transition modality becomes optional.
Further common actions (\ie \emph{b} and \emph{f}) are treated similarly under composition.
In contrast, transitions with
actions being exclusive to $Q$ or $D$ are preserved under composition.
As $Q\otimesforbidden D$ contains no illegal states, no
pruning is required in $P'=Q\parallelforbidden D$.
In contrast, assuming, \eg one of the inputs \emph{a}
of $Q$ being optional instead, then the initial state
of $P'$ would become illegal as $a\in O_D$, $p_D\oversetmay{a}$ and $p_Q\not\oversetmust{a}$,
and $Q$ and $D$ would be incompatible.

\begin{figure}[tp]
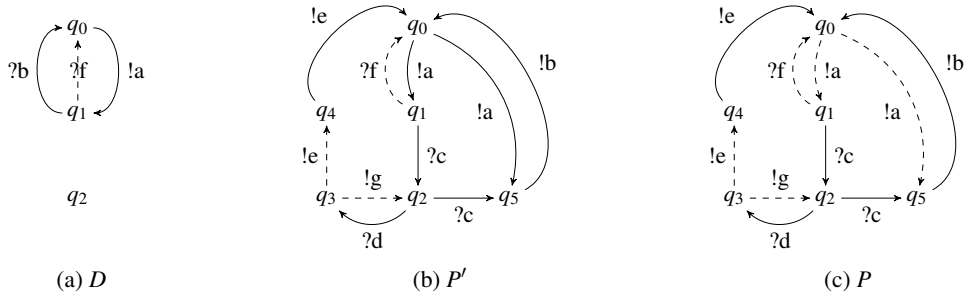

	\hfill
	\subfloat[$D$]{\label{fig:comp-example-d}\input{figures/decomp-example-divisor.tex}}
	\hfill
	\subfloat[$P'$]{\label{fig:comp-example-pprime}\input{figures/comp-example.tex}}
	\hfill
	\subfloat[$P$]{\label{fig:decomp-example}\input{figures/decomp-example-spec.tex}}
	\hfill\strut
	\caption{Example for Parallel Composition with Multicast and Quotienting (\cf Sect.~\ref{sec:decompositionality})}\label{fig:comp-example}
\end{figure}

We obtain the following compositionality
result for \textbf{modal-irioco}
with respect to parallel composition with multicast
communication.

\newcounter{compositionality-multicast-counter}
\setcounter{compositionality-multicast-counter}{\value{theorem}}
\begin{theorem}[Compositionality of \textbf{modal-irioco}]\label{theorem:compositionality-multicast}
	Let $s_1$, $s_2$, $i_1$, and $i_2$ be \miaforbidden with $i_1$ and $i_2$ being strong must-input-enabled, and $s_1$ and $s_2$ being compatible.
	Then it holds that
	$\left(i_1\mathmiocophi s_1\land i_2\mathmiocophi s_2\right)\Rightarrow i_1\parallelforbidden i_2\mathmiocophi s_1\parallelforbidden s_2$.
\end{theorem}

Theorem~\ref{theorem:compositionality-multicast}
is restricted to must-input-enabled implementations as
the input of an input/output pair has to be mandatory
(otherwise leading to an illegal state).
We further require \emph{strong} input-enabledness as
inputs in an input/output pair have to immediately react
to outputs (otherwise, again, leading to an illegal state).
Next, we show that IR-MIA parallel composition is \emph{associative}, thus facilitating multicast communication among multiple IR-MIA components being composed in arbitrary order.

\newcounter{parallel-composition-associativity-counter}
\setcounter{parallel-composition-associativity-counter}{\value{lemma}}
\begin{lemma}[Associativity of IR-MIA Parallel Composition]\label{lemma:parallel-composition-associativity}
	Let $P$, $Q$, $R$ be IR-MIA.
	It holds that $(P\parallelforbidden Q)\parallelforbidden R=P\parallelforbidden(Q\parallelforbidden R)$.
\end{lemma}

In addition, we show that compositionality of \textbf{modal-irioco} also holds if we combine multicast parallel composition with \emph{explicit hiding of outputs}, if specification $s$ has no $\tau$-steps.
For this, we first define \emph{parallel composition with hiding}.

\begin{definition}[IR-MIA Parallel Product and Composition with Hiding]\label{def:parallel-hiding}
	Two \miaforbidden $P_1$, $P_2$ are \textup{hiding composable (h-composable)} if $O_1 \cap O_2 = \emptyset$ and $I_1\cap I_2=\emptyset$.
	For such \miaforbidden we define the \textup{parallel product} $P_1\otimesforbiddenhiding P_2 = ((P_1\times P_2)\cup\{q_{12\Phi}\},I,O,\mustarrow,\mayarrow,q_{12\Phi})$, where $I=_\textit{def}(I_1\cup I_2)\setminus(O_1\cup O_2)$ and $O =_\textit{def} O_1\cup O_2$, and where $\mustarrow$ and $\mayarrow$ are the least relations satisfying the following conditions:
	\begin{tabbing}
		(May2/Must2)\quad \= $(p_1,p_2)\transrel{\alpha}_\gamma(p_1',p_2')$\quad \= if\quad \= \kill
		(May1/Must1) \> $(p_1,p_2)\transrel{\alpha}_\gamma (p_1',p_2)$ \> if \> $p_1\transrel{\alpha}_\gamma p_1'$ and $\alpha\notin A_2$ \\
		(May2/Must2) \> $(p_1,p_2)\transrel{\alpha}_\gamma (p_1,p_2')$ \> if \> $p_2\transrel{\alpha}_\gamma p_2'$ and $\alpha\notin A_1$ \\
		(May3/Must3) \> $(p_1,p_2)\transrel{\tau}_\gamma (p_1',p_2')$ \> if \> $p_1\transrel{a}_\gamma p_1'$ and $p_2\transrel{a}_\gamma p_2'$ for some $a$.
	\end{tabbing}
	From this parallel product with hiding, we obtain the \textup{parallel composition with hiding} $P_1\hidingforbidden P_2$ by the same pruning procedure as in Def.~\ref{def:parallel-composition}.
\end{definition}

We obtain the following compositionality result for \textbf{modal-irioco} with respect to parallel composition with hiding.

\newcounter{compositionality-hiding-counter}
\setcounter{compositionality-hiding-counter}{\value{theorem}}
\begin{theorem}[Compositionality of \miocophi Regarding Parallel Composition with Hiding]\label{theorem:compositionality-hiding}
	Let $s_1$, $s_2$, $i_1$, and $i_2$ be strongly must-input-enabled \miaforbidden.
	Then $\left(i_1\mathmiocophi s_1\land i_2\mathmiocophi s_2\right)\Rightarrow i_1\hidingforbidden i_2\mathmiocophi s_1\hidingforbidden s_2$ if $s_1$ and $s_2$ are compatible, $\forall q\in Q_{s1}:\forall i\in I_{s1}\cap O_{s2}:q\not\oversetmay{i}q_{s1\Phi}$, and $\forall q\in Q_{s2}:\forall i\in I_{s2}\cap O_{s1}:q\not\oversetmay{i}q_{s2\Phi}$.
\end{theorem}

Similar to parallel composition with multicast, parallel composition with hiding is also associative with the restriction that $P$, $Q$ and $R$ do not synchronize on the same actions.

\newcounter{parallel-composition-hiding-associativity-counter}
\setcounter{parallel-composition-hiding-associativity-counter}{\value{lemma}}
\begin{lemma}[Associativity of IR-MIA Parallel Composition]\label{lemma:parallel-composition-hiding-associativity}
	Let $P$, $Q$, $R$ be IR-MIA.
	It holds that $(P\parallelforbidden Q)\parallelforbidden R=P\parallelforbidden(Q\parallelforbidden R)$ if pairwise intersection of $I_P\cap O_Q$ and $I_Q\cap O_P$ with $I_Q\cap O_R$ and $I_R\cap O_Q$ results in $\emptyset$.
\end{lemma}
%
\section{Decompositionality}\label{sec:decompositionality}

Compositionality of \textbf{modal-irioco}
allows for decomposing I/O conformance testing
of systems consisting of several interacting components.
In particular, given two components
$c_1$, $c_2$ being supposed to implement corresponding specifications
$s_1$, $s_2$, then Theorem~\ref{theorem:compositionality-multicast}
ensures that if $c_1 \mathmiocophi s_1$ and $c_2 \mathmiocophi s_2$ holds, then
$c_1 \parallelforbidden c_2 \mathmiocophi s_1 \parallelforbidden s_2$
is guaranteed without the need for (re-)testing after composition.
However, in order to benefit from this property,
a mechanism is required to decompose specifications
$s = s_1 \parallelforbidden s_2$ and respective implementations
$i=c_1 \parallelforbidden c_2$, accordingly.
Interface theories therefore provide quotient operators $\sslash$
serving as the inverse to parallel composition
(\ie if $c_1 \parallel c_2 = c$ then $c\sslash c_1 = c_2$), where
$c_2$ is often referred to as \emph{unknown component}~\cite{Villa2011} or \emph{testing context}~\cite{Noroozi2013}.
We therefore adopt the quotient operator defined for MIA with universal state~\cite{Bujtor2015a}
to IR-MIA.
Similar to parallel composition, the quotient
operator is defined in two steps.
\begin{enumerate}
	\item The \emph{pseudo-quotient} $P\oslash D$
is constructed as appropriate communication partner (if exists)
for a given divisor $D$ with respect to the overall specification $P$.
	\item The \emph{quotient} $P\sslashforbidden D$ is derived
from $P\oslash D$, again, by pruning erroneous states.
\end{enumerate}
For this, we require $P$ and $D$ to be $\tau$-free and $D$ to be {\em may-deterministic}
(\ie $d\oversetmay{a} d'$ and $d\oversetmay{a} d''$ implies $d'=d''$).
In contrast to \cite{Bujtor2015a}, we restrict
our considerations to IR-MIA with at least one
state and one may-transition.
A pair $P$ and $D$ satisfying
these restrictions is called a {\em quotient pair}.

\begin{definition}[IR-MIA Pseudo-Quotient]\label{def:pseudo-quotient}
	Let $(P,I_P,O_P,\mustarrow,\mayarrow,p_\Phi)$ and $(D,I_D,O_D,\mustarrow,\mayarrow,d_\Phi)$ be a \miaforbidden quotient pair with $A_D\subseteq A_P$ and $O_D\subseteq O_P$.
	We set $I=_\textit{def}I_P\cup O_D$ and $O=_\textit{def}O_P\setminus O_D$.
	$P\oslash D=_\textit{def}(P\times D,I,O,\mustarrow,\mayarrow,(p_\Phi,d_\Phi))$, where the transition relations are defined by the rules:
	\begin{tabbing}
		(QMay7/QMust7)\; \= $(p,d)\oversetgamma{a}(p_\Phi,d_\Phi)$\; \= if\; \= \kill
		(QMay1/QMust1) \> $(p,d)\oversetgamma{a}(p',d)$ \> if \> $p\oversetgamma{a}p'\neq p_\Phi$ and $a\notin A_D$ \\
		(QMay2) \> $(p,d)\oversetmay{a}(p',d')$ \> if \> $p\oversetmay{a}p'\neq p_\Phi$ and $d\oversetmust{a}d'\neq d_\Phi$ \\
		(QMay3) \> $(p,d)\oversetmay{a}(p',d')$ \> if \> $p\oversetmay{a}p'\neq p_\Phi$, $d\oversetmay{a}d'\neq d_\Phi$ and $a\notin O_P\cap I_D$ \\
		(QMust2) \> $(p,d)\oversetmust{a}(p',d')$ \> if \> $p\oversetmust{a}p'\neq p_\Phi$ and $d\oversetmust{a}d'\neq d_\Phi$ \\
		(QMust3) \> $(p,d)\oversetmust{a}(p',d')$ \> if \> $p\oversetmay{a}p'\neq p_\Phi$, $d\oversetmay{a}d'\neq d_\Phi$ and $a\in O_D$ \\
		(QMay4/QMust4) \> $(p,d)\oversetgamma{a}(p_\Phi,d_\Phi)$ \> if \> $p\oversetgamma{a}p_\Phi$ and $d\not\oversetmust{a}d_\Phi$.
	\end{tabbing}
\end{definition}

The Rules \emph{(QMay1/QMust1)} to \emph{(QMust3)} require $p\neq p_\Phi$, as the special case $p = p_\Phi$ is handled by rule \emph{(QMay4/QMust4)}.
Rule \emph{(QMay1/QMust1)} concerns transitions with uncommon actions.
Rule \emph{(QMay2)} requires a mandatory transition with action in $D$ as composition requires input transitions labeled with common actions to be mandatory (the additional requirement of Rule \emph{(QMay3)} is stated for the same reason).
Rule \emph{(QMust3)} only requires transitions to be optional, because if $a\in O_D$ holds, then the resulting transition accepts as input a common action (which must be mandatory for the composition).

The quotient $P\sslashforbidden D$ is derived
from pseudo-quotient $P\oslash D$ by recursively pruning
all so-called \emph{impossible states} $(p,d)$ (\ie states
leading to erroneous parallel composition).

\begin{definition}[IR-MIA Quotient]\label{def:quotient}
	The set $G\subseteq P\times D$ of \emph{impossible states}
	of pseudo-quotient $P\oslash D$ is defined as the least set satisfying the rules:
	\begin{tabbing}
		(G4)\quad \= $p\oversetmust{a}p'\neq p_\Phi$ and $d\not\oversetmust{a}d_\Phi$ and $a\in A_D$\quad \= implies\quad \= \kill
		(G1) \> $p\oversetmust{a}p'\neq p_\Phi$ and $d\not\oversetmust{a}$ and $a\in A_D$ \> implies \> $(p,d)\in G$ \\
		(G2) \> $p\oversetmust{a} p_\Phi$ and $d\oversetmay{a}$ and $a\in O_D$ \> implies \> $(p,d)\in G$ \\
		(G3) \> $(p,d)\oversetmust{a}r$ and $r\in G$ \> implies \> $(p,d)\in G$.
	\end{tabbing}
	The \emph{quotient} $P\sslashforbidden D$ is obtained from $P\oslash D$
	by deleting all states $(p,d)\in G$ (and respective transitions).
	If $(p,d)\in P\sslashforbidden D$,
	then we write $p\sslashforbidden d$, and
	quotient $P\sslashforbidden D$ is defined.
\end{definition}

Rule \emph{(G1)} ensures that for a transition labeled with a common action, there is a corresponding transition in the divisor (otherwise, the state is \emph{impossible} and therefore removed).
Rule \emph{(G2)} ensures that a forbidden action of the specification is also forbidden in the divisor (otherwise, the state is considered \emph{impossible}).
Finally, Rule \emph{(G3)} (recursively) removes all states from which \emph{impossible} states are reachable.

For example, consider the quotient $Q=P\sslashforbidden D$
(\cf Fig.~\ref{fig:ir-mia-spec}, Fig.~\ref{fig:comp-example-d}, and Fig.~\ref{fig:decomp-example}).
A common action becomes input action
in $Q$ if it is an input action in both $P$
and $D$ (\eg \textit{f}), and likewise for output actions.
If a common action is output action of $P$ and input action of $D$,
then it becomes output of $Q$ (\eg \textit{b}).
In contrast, a common action must not be
input action of $P$ and output action of $D$ as composing
outputs with inputs always yields outputs.
Actions being exclusive to $P$ are treated
similar to parallel composition, whereas $D$ must not have
exclusive actions (\cf Def.~\ref{def:pseudo-quotient}).

\begin{figure}[tp]
	\centering
	\input{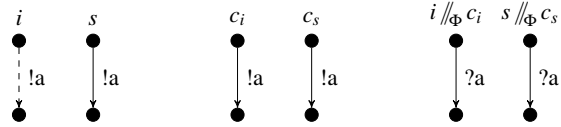}
	\caption{Example for the necessity of mandatory outputs for $i'\mathmiocophi s'\Rightarrow i\mathmiocophi s$ of Theorem~\ref{theorem:decompositionality}.}
	\label{fig:decompositionality}
\end{figure}

For \emph{decomposability} to hold for \textbf{modal-irioco}
(\ie $i\sslashforbidden c_i\mathmiocophi s\sslashforbidden c_s\land c_i\mathmiocophi c_s\Rightarrow i\mathmiocophi s$ ),
we further require $i$ to only have mandatory outputs
as illustrated in Fig.~\ref{fig:decompositionality}: here,
$i\mathmiocophi s$ does not hold, although
$c_i\mathmiocophi c_s$ and $i\sslashforbidden c_i\mathmiocophi s\sslashforbidden c_s$ holds.
This is due to the fact that optional outputs
combined with mandatory outputs become mandatory inputs in the quotient
(as parallel composition requires inputs of an input/output pair to be mandatory).
The following result ensures
that the quotient operator on IR-MIA indeed
serves (under the aforementioned restrictions)
as the inverse to parallel composition
with respect to \textbf{modal-irioco}.

\newcounter{decompositionality-counter}
\setcounter{decompositionality-counter}{\value{theorem}}
\begin{theorem}[Decompositionality of \textbf{modal-irioco}]\label{theorem:decompositionality}
	Let $i$, $s$, $c_i$, and $c_s$ be \miaforbidden with
	$i$ and $c_i$ being weak must-input-enabled and all output behaviors of $i$ being mandatory.
	Then $i\mathmiocophi s$ if
	$i\sslashforbidden c_i\mathmiocophi s\sslashforbidden c_s$ and $c_i\mathmiocophi c_s$.
\end{theorem}

Based on this result, \textbf{modal-irioco} supports
\emph{synthesis of testing environments} for
\emph{testing through contexts}~\cite{Noroozi2013,Daca2014}, as well as a solution to the \emph{unknown-component} problem~\cite{Villa2011}.
%
\section{Case Study}\label{sec:casestudy}

In this section, we present a small case with real-world examples for \miocophi (including negative testing capabilities), \miaforbidden refinement, parallel composition with multicast and hiding, and quotienting.

First, Figs.~\ref{fig:case-study-mioco-impl} and~\ref{fig:case-study-mioco-impl} give an example for \miocophi with a simple vending machine.
Here, the specification $s$ accepts \emph{2\euro{}} as mandatory input and \emph{1\euro{}} as optional input.
When \emph{2\euro{}} are entered, \emph{change} is returned.
Afterwards, the user may choose between \emph{coffee} and \emph{tea}, or \emph{cups} may be refilled.
If the user chooses \emph{tea}, then either a \emph{cup} of tea (mandatory) or an \emph{error} message (optional) is returned.
If the user chooses \emph{coffee}, then either a \emph{cup} of coffee or an \emph{error} message is returned (both optional).
The implementation $i$ is similar to $s$ with the differences being a missing \emph{cup} output after \emph{coffee} and forbidden inputs of \emph{1\euro{}} at the initial state and after entering \emph{2\euro{}}.
For this example, it holds that $i\mathmiocophi s$:
The output of a \emph{cup} after \emph{coffee} may be removed because the output is optional.
Additionally, the input of \emph{1\euro{}} at the initial state may be forbidden because the input is optional as well.
Furthermore, the input of \emph{1\euro{}} after \emph{2\euro{}} may be forbidden because the input is unspecified in $s$.
By forbidding the additional input of \emph{1\euro{}} after \emph{2\euro{}}, we show the negative testing capabilities of \miocophi as now, no variant $i'\miarefforbidden i$ of the implementation is allowed to perform critical unspecified behavior.
Otherwise, a variant could, \eg return an unlimited amount of tea after entering \emph{2\euro{}} followed by \emph{1\euro{}}.
By utulizing the failure state, this is not possible anymore.

\begin{figure}[tp]
	\hfill
	\subfloat[Refinement $i'\miarefforbidden i$]{\label{fig:case-study-mioco-ref}\input{figures/case-study-mioco-ref.tex}}
	\hfill
	\subfloat[Implementation $i$]{\label{fig:case-study-mioco-impl}\input{figures/case-study-mioco-impl.tex}}
	\hfill
	\subfloat[Specification $s$]{\label{fig:case-study-mioco-spec}\input{figures/case-study-mioco-spec.tex}}
	\hfill\strut
	\caption{Example for \miocophi and \miaforbidden refinement, where $i\mathmiocophi s$ and $i'\miarefforbidden i$.}\label{fig:case-study-mioco}
\end{figure}

Second, Figs.~\ref{fig:case-study-mioco-ref} and~\ref{fig:case-study-mioco-impl} give an example for \miaforbidden refinement.
Here, it holds that $i'\miarefforbidden i$ because the optional output of an \emph{error} message after \emph{coffee} becomes mandatory in $i'$, and the optional \emph{error} message after \emph{tea} is removed.

Third, Fig.~\ref{fig:composition-example} gives an example (adapted from de Alfaro and Henzinger~\cite{Alfaro2005}) for the \miaforbidden parallel product and parallel composition where two automata $P$ (Fig.~\ref{fig:composition-example-p}) and $Q$ (Fig.~\ref{fig:composition-example-q}) are composed.
$P$ is capable of entering \emph{1\euro{}}, and then it waits for a \emph{cup} from the vending machine.
Note that the set of inputs of $P$ contains \emph{retry} but there is no state accepting that input.
$Q$ waits for another automaton to enter \emph{1\euro{}}.
After that, the user can choose the \emph{size} of the beverage and the beverage itself (\emph{coffee} or \emph{tea}).
If \emph{tea} is chosen, the machine returns a \emph{cup}.
If \emph{coffee} is chosen, $Q$ performs a \emph{reset} followed by a \emph{retry} (unfortunately, this vending machine is hostile towards coffee drinkers).
Therefore, the common actions of $P$ and $Q$ are \emph{1\euro}, \emph{cup}, and \emph{retry}.

\begin{figure}[tp]
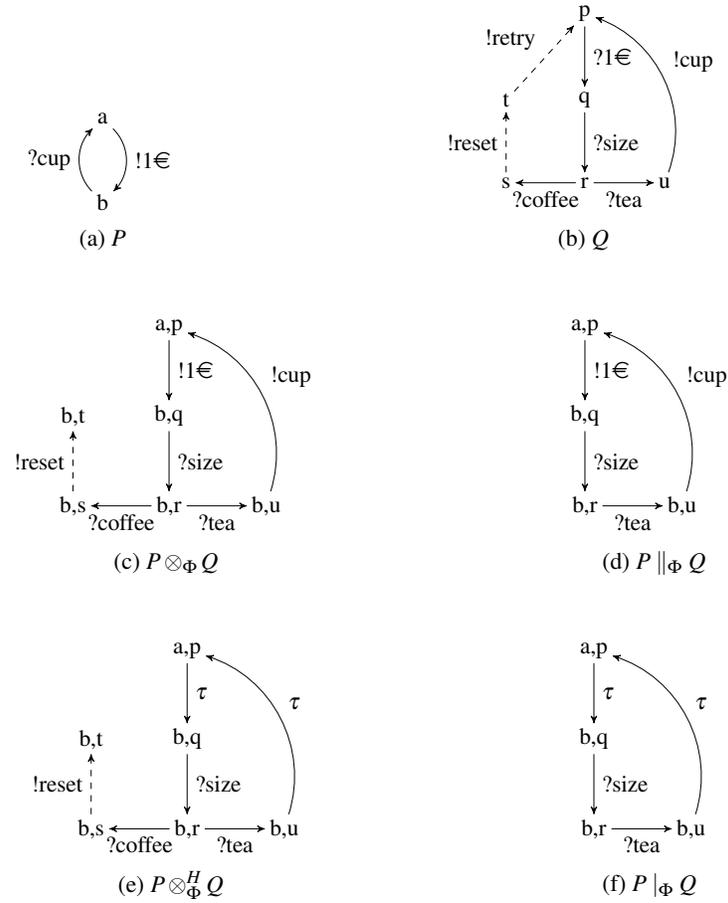

	\hfill
	\subfloat[$P$]{\label{fig:composition-example-p}\input{figures/composition-example-p.tex}}
	\hfill
	\subfloat[$Q$]{\label{fig:composition-example-q}\input{figures/composition-example-q.tex}}
	\hfill\strut

	\bigskip

	\hfill
	\subfloat[$P\otimesforbidden Q$]{\label{fig:composition-example-multicast-product}\input{figures/composition-example-multicast-product.tex}}
	\hfill
	\subfloat[$P\parallelforbidden Q$]{\label{fig:composition-example-multicast-final}\input{figures/composition-example-multicast-final.tex}}
	\hfill\strut

	\bigskip

	\hfill
	\subfloat[$P\otimesforbiddenhiding Q$]{\label{fig:composition-example-hiding-product}\input{figures/composition-example-hiding-product.tex}}
	\hfill
	\subfloat[$P\hidingforbidden Q$]{\label{fig:composition-example-hiding-final}\input{figures/composition-example-hiding-final.tex}}
	\hfill\strut

	\caption{Parallel composition of \miaforbidden with multicast and hiding, where $I_P=\{\text{cup, retry}\}$, $O_P=\{\text{1\euro}\}$, $I_Q=\{\text{coffee, size, tea}\}$, and $O_Q=\{\text{cup, reset, retry}\}$.
	The example is adapted from de Alfaro and Henzinger~\cite{Alfaro2005}.}\label{fig:composition-example}
\end{figure}

In order to obtain the parallel product with multicast $P\otimesforbidden Q$ depicted in Fig.~\ref{fig:composition-example-multicast-product}, we first combine the states \emph{a} and \emph{p}.
The only action of both states is the common action \emph{1\euro{}}, so this action is performed synchronously.
\emph{1\euro{}} remains an output action in the parallel product because there may be other automata also receiving that action (\eg a component counting money).
Afterwards, state \emph{b} of $P$ has to wait for $Q$ because the only action of \emph{b} (\emph{?cup}) is a common action.
After that, $Q$ performs \emph{?size}, \emph{?coffee}, \emph{?tea}, and \emph{!reset} independently because the states \emph{q}, \emph{r}, and \emph{s} do not have any outgoing transitions with common actions of $P$ and $Q$.
The state \emph{t} cannot perform any actions because \emph{retry} is a common action of $P$ and $Q$ but \emph{b} is not able to receive that action.
However, \emph{b} and \emph{u} can both perform \emph{cup}, and then the parallel product has a transition with $a,p$ (the initial state) as its target because in both $P$ and $Q$ the action leads to the initial state.

To obtain the parallel composition with multicast $P\parallelforbidden Q$ (\cf Fig.~\ref{fig:composition-example-multicast-final}), we need to find and remove all illegal states of $P\otimesforbidden Q$.
Initially, there is one illegal state in the parallel product (\cf Fig.~\ref{fig:composition-example-multicast-product}).
The component \emph{t} of the state \emph{(b,t)} is able to perform the output \emph{retry} but \emph{b} cannot perform that action although it is in the set of actions of $P$.
Additionally, there are no actions for \emph{b} or \emph{t} that can be performed without the other component.
Therefore, \emph{(b,t)} of the parallel product is an illegal state.
There are no other initial illegal states so next we look for states which may reach the illegal state autonomously, \ie through output and internal actions.
This leads us to \emph{(b,s)} being able to reach \emph{(b,t)} through \emph{!reset}.
Therefore, \emph{(b,s)} is also added to the set of illegal states.
After that, there are no additional illegal states left.
The last step consists of deleting all illegal states (\ie \emph{(b,s)} and \emph{(b,t)}) and all their outgoing and incoming transitions (\ie \emph{!reset} and \emph{?coffee}).
In that way, we generated the parallel composition $P\parallelforbidden Q$.
The parallel product and parallel composition with hiding (\cf Figs.~\ref{fig:composition-example-hiding-product} and~\ref{fig:composition-example-hiding-final}) are obtained similar to the parallel composition with hiding.
The only difference are common actions becoming internal actions $\tau$ instead of output actions.

Finally, Fig.~\ref{fig:case-study-decomposition} gives an example for quotienting.
Here, specification $P$ and divisor $D$ are similar to the \miaforbidden depicted in Figs.~\ref{fig:composition-example-multicast-final} and~\ref{fig:composition-example-p} with the difference that \emph{1\euro} is an optional action in both automata.
The quotient $Q=P\sslashforbidden D$ is obtained by performing the inverse of parallel composition with multicast, \ie the quotient of an output in both specification and divisor is an input action, the quotient of an output of the specification and an input of the divisor is an output action, and the quotient of an input in both specification and divisor is an input.
Therefore, we first build the quotient of the states \emph{(a,p)} and \emph{a}.
Here, \emph{1\euro{}} becomes a mandatory input because for parallel composition, the input of a common actions must always be mandatory.
Afterwards, \emph{size} and \emph{tea} are copied to $Q$ because both actions are not common actions.
In the last step, the common action \emph{cup} becomes an output in $Q$ because \emph{cup} is an output in $P$ and an input in $D$.
In that way, we generated the quotient $Q=P\sslashforbidden D$.

\begin{figure}[tp]
	\hfill
	\subfloat[Specification $P$]{\label{fig:case-study-decomposition-p}\input{figures/case-study-decomposition-p.tex}}
	\hfill
	\subfloat[Divisor $D$]{\label{fig:case-study-decomposition-d}\input{figures/case-study-decomposition-d.tex}}
	\hfill
	\subfloat[Quotient $Q=P\sslashforbidden D$]{\label{fig:case-study-decomposition}\input{figures/case-study-decomposition-q.tex}}
	\hfill\strut
	\caption{}\label{}
\end{figure}

\section{Conjunction and Disjunction}\label{sec:conjunction-disjunction}
The MIA theory includes operators 
for the conjunction and disjunction of two given MIA, where
the conjunction corresponds to the \emph{infimum}, and
the disjunction corresponds to the \emph{suprenum} of those
MIA with respect to the partial refinement preorder relation $\sqsubseteq$.
In this section, we investigate both operators
in the context of IR-MIA and \miocophi{}.
\subsection{Conjunction}\label{subsec:conjunction}
The binary \emph{conjunction} of two MIA, $P$ and $Q$, constructs 
a MIA $P\land Q$ comprising all variants shared by $P$ and $Q$.
An optional transition in $P$ may be refined in variants $P'\sqsubseteq P$ 
to either remain optional, or to become mandatory or forbidden behavior, 
whereas a mandatory transition in $P$ must remain mandatory in $P'$.
Hence, under conjunction, transitions being optional (mandatory) in both
$P$ and $Q$ remain optional (mandatory) in $P\land Q$.
In contrast, a transition being optional in $P$, but mandatory
in $Q$ (or, vice versa), becomes mandatory in $P\land Q$.
Beyond this intuitive construction, 
additional rules must be adapted to the
original MIA conjunction in order to be applicable to IR-MIA.

First, in the original MIA theory~\cite{Bujtor2015a}, the notion of \emph{disjunctive transitions}
is introduced, defining an extension to the expressiveness of MIA refinement
being crucial for both the construction of conjunction and disjunction.
Syntactically, a disjunctive transition is a must-transition with a single source state, but 
with a \emph{set} of possible target states.
Semantically, a disjunctive transition defines an \emph{inclusive or}
among the different target states (\ie{} \emph{at least one} of the 
may-transitions underlying a disjunctive transition must be preserved under MIA refinement).
We may easily extend our definition of IR-MIA, accordingly, to also
provide disjunctive transitions as follows.
\begin{definition}[IR-MIA with Disjunctive Transitions]\label{def:mia-disjunctive}
	An \emph{IR-MIA with disjunctive transitions} is an IR-MIA according to 
  Def.~\ref{def:ir-mia}, where $\mustarrow\subseteq ((Q\setminus\{q_\Phi\}) \times I_Q \times \mathcal{P}(Q)) \cup ((Q\setminus\{q_\Phi\}) \times \mathcal{P}(O_Q\cup\{\tau\}) \times (Q\setminus\{q_\Phi\}))$.
\end{definition}
Figure~\ref{fig:disjunctive-transitions-mioco-i} shows a sample IR-MIA $i$ with 
a disjunctive transition.
\begin{figure}[tp]
	\hfill
	\subfloat[$i$]{\label{fig:disjunctive-transitions-mioco-i}\input{figures/disjunctive-transitions-mioco-i.tex}}
	\hfill
	\subfloat[$i'\miarefforbidden i$]{\label{fig:disjunctive-transitions-mioco-ip}\input{figures/disjunctive-transitions-mioco-ip.tex}}
	\hfill
	\subfloat[$i''\miarefforbidden i$]{\label{fig:disjunctive-transitions-mioco-ipp}\input{figures/disjunctive-transitions-mioco-ipp.tex}}
	\hfill
	\subfloat[$s$]{\label{fig:disjunctive-transitions-mioco-s}\input{figures/disjunctive-transitions-mioco-s.tex}}
	\hfill\strut
	\caption{Example for disjunctive transitions and the incompatability with \miocophi{}.}\label{fig:disjunctive-transitions-mioco}
\end{figure}
Here, the initial state accepts input \emph{i} and has two 
possible target states for this input, where each variant 
must include at least one of these two transitions.
Hence, the IR-MIA depicted in Figures~\ref{fig:disjunctive-transitions-mioco-i}, 
\ref{fig:disjunctive-transitions-mioco-ip}, and 
\ref{fig:disjunctive-transitions-mioco-ipp} are all
valid refinements of $s$.
The adapted refinement relation on IR-MIA with
disjunctive transitions may be defined as follows.
\begin{definition}[IR-MIA Refinement with Disjunctive Transitions]\label{def:mia-ref-disjunctive}
	Let $P$, $Q$ be \miaforbidden{} with $I_P=I_Q$ and $O_P=O_Q$.
	The \emph{IR-MIA refinement relation with disjunctive transitions} $\mathcal{R}\subseteq P\times Q$ is defined according to Def.~\ref{def:mia-refinement-forbidden} where clauses~\ref{refine-must-i} and~\ref{refine-must-o} are replaced by
	\begin{enumerate}
		\setcounter{enumi}{\value{mia-ref-must-i}}
		\item $q\oversetmust{i}Q'\setminus\{q_\Phi\}$ implies $\exists P'.p\oversetmust{i}\Oversetmust{\epsilon}P'\setminus\{p_\Phi\}$ and $\forall p'\in P'\exists q'\in Q'.(p',q')\in\mathcal{R}$, and
		\setcounter{enumi}{\value{mia-ref-must-o}}
		\item $q\oversetmust{\omega}Q'$ implies $\exists P'.p\Oversetmust{\hat\omega}P'$ and $\forall p'\in P'\exists q'\in Q'.(p',q')\in\mathcal{R}$.
	\end{enumerate}
\end{definition}
In particular, disjunctive transitions are required during IR-MIA conjunction whenever
one IR-MIA contains a non-deterministic choice over some action and the other IR-MIA
contains a must-transition with the same action.
Figure~\ref{fig:conjunction-disjunctive-trans} shows an example for a conjunction $P\land Q$ 
necessarily resulting in a disjunctive transition.
\begin{figure}[tp]
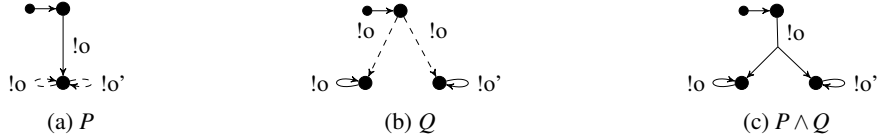

	\hfill
	\subfloat[$P$]{\label{fig:conjunction-disjunctive-trans-p}\input{figures/conjunction-disjunctive-trans-p.tex}}
	\hfill
	\subfloat[$Q$]{\label{fig:conjunction-disjunctive-trans-q}\input{figures/conjunction-disjunctive-trans-q.tex}}
	\hfill
	\subfloat[$P\land Q$]{\label{fig:conjunction-disjunctive-trans-p-and-q}\input{figures/conjunction-disjunctive-trans-p-and-q.tex}}
	\hfill\strut
	\caption{Example for a disjunctive transition being generated by applying conjunction (adapted from~\cite{Bujtor2014}).}\label{fig:conjunction-disjunctive-trans}
\end{figure}
$P$ has mandatory output \emph{o}, followed by optional outputs \emph{o} and \emph{o'}, whereas
$Q$ has a non-deterministic choice between two optional transitions both with output \emph{o}, followed 
either by the mandatory output \emph{o}, or \emph{o'}.
Hence, $P\land Q$ must include a disjunctive transition as, otherwise, $P\land Q$ 
cannot choose between providing either \emph{o} or \emph{o'}, or both of them.
Using two optional non-deterministic outputs \emph{o} instead
would permit a variant without any output \emph{o}, whereas 
using two mandatory non-deterministic output transitions 
would not permit any variant with only one possible output after \emph{o}.
In order to apply \miocophi{} to IR-MIA with disjunctive transitions,
we have to adapt $\mustafter$ as follows.
\begin{definition}[\miocod{}]
Let $Q$ be an IR-MIA with disjunctive transitions and $p\in Q$.
We define $p\mustafterd\sigma:=\{p'\mid p\Oversetmust{\sigma}P',p'\in P'\}$.
By \miocod{}, we define \textbf{modal-irioco} according to Def.~\ref{def:mioco} 
where $\mustafter$ is replaced by $\mustafterd$.
\end{definition}
Unfortunately, disjunctive transitions obstruct the preservation of \miocophi{} under 
refinement (\cf{}~Theorem~\ref{theorem:mioco-refinement-compatability}).
\begin{lemma}\label{lemma:mioco-disjunctive-transitions}
Let $i$ and $s$ be IR-MIA with disjunctive transitions.
Then it holds that $i\mathmiocod s\not\Rightarrow \forall i'\miarefforbidden i:\exists s'\miarefforbidden s:i'\mathmiocod s'$.
\end{lemma}
In particular, due to the purely trace-based nature of conformance testing in general and \miocophi{} in particular, 
it is not possible to distinguish occurrences of actions related to 
must-transitions being part of disjunctive transitions from actions
related to singleton must-transitions.
However, the former may be removed under modal refinement, whereas the latter 
must be preserved.
An example for illustrating this issue is provided in Figure~\ref{fig:disjunctive-transitions-mioco}.
Let $i$, $i'$, $i''$, and $s$ be \miaforbidden{} with disjunctive transitions.
Here, it holds that $i\mathmiocophi s$.
However, neither $i'\mathmiocophi s$ nor $i''\mathmiocophi s$ holds.
As a consequence, \miocophi{} and disjunctive transitions are incompatible.
Therefore, we restrict IR-MIA and conjunction on IR-MIA,
as compared to the original MIA theory~\cite{Bujtor2015a},
in order to avoid any occurrence of disjunctive transitions in the following.

A further necessary adaptation of MIA conjunction to IR-MIA
results from the possibility of underspecification 
(\ie{} an unspecified input may be implemented arbitrarily).
To this end, we introduce a special \emph{demonic state} $p_d$, 
serving as target state for previously unspecified inputs.
The demonic state has no outgoing input transition (\ie{} all inputs are 
unspecified) and optional self-transitions for every possible output.
\begin{definition}[Demonic State]\label{def:demonic state}
	A state $p_d$ of a \miaforbidden{} 
  is a \emph{demonic state} 
  if $\forall o\in O:p_d\oversetmay{o}p_d$ and $\forall a\in I\cup\{\tau\}:p_d\not\oversetmay{a}$.
\end{definition}
Consider the IR-MIA depicted in Figure~\ref{fig:demonic-state-example}
with optional output \emph{o} and optional input \emph{i}.
\begin{figure}[tp]
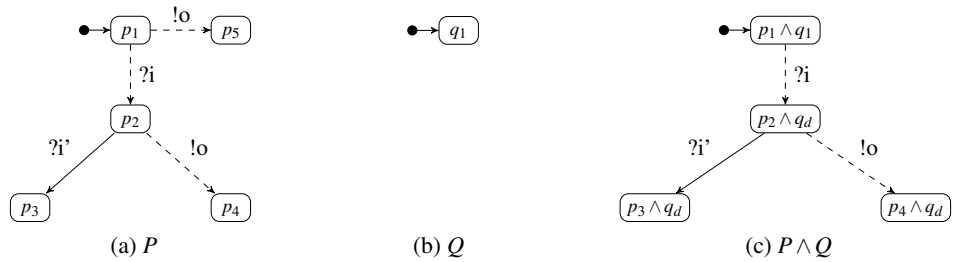

	\hfill
	\subfloat[$P$]{\label{fig:demonic-state-example-p}\input{figures/demonic-state-example-p.tex}}
	\hfill
	\subfloat[$Q$]{\label{fig:demonic-state-example-q}\input{figures/demonic-state-example-q.tex}}
	\hfill
	\subfloat[$P\land Q$]{\label{fig:demonic-state-example-p-and-q}\input{figures/demonic-state-example-p-and-q.tex}}
	\hfill\strut
	\caption{Example for the demonic state in conjunction.}\label{fig:demonic-state-example}
\end{figure}
After input \emph{i}, $P$ accepts the mandatory input \emph{i'} and has the optional output \emph{o}, whereas
$Q$ has no outgoing transitions.
Hence, the conjunction of an optional input \emph{i} in $P$ and an unspecified input \emph{i} 
in $Q$ should result in an optional input \emph{i} in $P\land Q$ 
as unspecified inputs may be implemented arbitrarily.
This is achieved by setting the target of \emph{i} to $p_2$ and the demonic state $q_d$ in $P\land Q$.
The state $p_2\land q_d$ has outgoing transitions for \emph{i'} and \emph{o} as these 
actions are specified in $P$ and in $Q$ and arbitrary subsequent 
behavior is allowed after an unspecified input action.
This example shows that $q_d$ must have optional outgoing transitions for 
every possible output.
We are now able to define conjunction on IR-MIA in two consecutive steps: 
(1) the conjunctive product $P_1\&P_2$ on two \miaforbidden{}, $P_1$ and $P_2$, followed by
(2) the conjunction $P_1\land P_2$ by removing erroneous state pairs $(p_1,p_2)$ from $P_1\&P_2$.
Concerning step (1), we first require $P_1$ and $P_2$ to have similar alphabets.
In $P_1 \& P_2$, a fresh state $p_{12\Phi}$ is introduced to serve as 
unique failure state.
The modality $\gamma$ of composed transitions 
$(p_1,p_2)\transrel{\alpha}_\gamma (p_1',p_2')$ depends on the modality of the individual transitions.
\begin{definition}[IR-MIA Conjunctive Product]\label{def:conjunctive-product}
	Consider two \miaforbidden $(P_1,I,O,\mustarrowsub{P1},\mayarrowsub{P1},p_{1\Phi})$ and $(P,I,O,\mustarrowsub{P2},\mayarrowsub{P2},p_{2\Phi})$ with common alphabets.
	The \emph{conjunctive product} is defined as $P_1\&P_2=_\textit{def}((P_1\cup\{p_{1d}\})\times (P_2\cup\{p_{2d}\})\},I,O,\mustarrow,\mayarrow,p_{12\Phi})$ with $\mustarrow,\mayarrow$ being the smallest relations derived by the following operational rules:
	\begin{tabbing}
		(OMust2)\quad \= $(p_1,p_2)\oversetmust{\omega}(p_1',p_2')$\quad \= if\quad \= \kill
		(OMust1) \> $(p_1,p_2)\oversetmust{\omega}(p_1',p_2')$ \> if \> $p_1\oversetmust{\omega}p_1'$ and $p_2\Oversetmay{\hat\omega}$ \\
		(OMust2) \> $(p_1,p_2)\oversetmust{\omega}(p_1',p_2')$ \> if \> $p_1\Oversetmay{\hat\omega}$ and $p_2\oversetmust{\omega}p_2'$ \\
		(IMust1) \> $(p_1,p_2)\oversetmust{i}(p_1',p_2')$ \> if \> $p_1\oversetmust{i}p_1'\neq p_{1\Phi}$ and $p_2\oversetmay{i}\Oversetmay{\epsilon}p_2'\neq p_{2\Phi}$ \\
		(IMust2) \> $(p_1,p_2)\oversetmust{i}(p_1',p_2')$ \> if \> $p_1\oversetmay{i}\Oversetmay{\epsilon}p_1'\neq p_{1\Phi}$ and $p_2\oversetmust{i}p_2'\neq p_{2\Phi}$ \\
		(OMay) \> $(p_1,p_2)\oversetmay{\omega}(p_1',p_2')$ \> if \> $p_1\Oversetmay{\hat\omega}p_1'$ and $p_2\Oversetmay{\hat\omega}p_2'$ \\
		(IMay) \> $(p_1,p_2)\oversetmay{i}(p_1',p_2')$ \> if \> $p_1\oversetmay{i}\Oversetmay{\epsilon}p_1'\neq p_{1\Phi}$ and $p_2\oversetmay{i}\Oversetmay{\epsilon}p_2'\neq p_{2\Phi}$ \\
		(DMay1) \> $(p_1,p_2)\oversetgamma{i}(p_1',p_{2d})$ \> if \> $p_1\oversetgamma{i}p_1'\neq p_{1\Phi}$ and $p_2\not\oversetmay{i}$ \\
		(DMay2) \> $(p_1,p_2)\oversetgamma{i}(p_{1d},p_2')$ \> if \> $p_1\not\oversetmay{i}$ and $p_2\oversetgamma{i}p_2'\neq p_{2\Phi}$ \\
		(FMust) \> $(p_1,p_2)\oversetmust{i}p_{12\Phi}$ \> if \> $p_1\oversetmust{i}p_{1\Phi}$ or $p_2\oversetmust{i}p_{2\Phi}$ \\
		(FMay1) \> $(p_1,p_2)\oversetmay{i}p_{12\Phi}$ \> if \> $p_1\oversetmay{i}$ and $p_2\oversetmay{i}p_{2\Phi}$ \kill
		(FMay2) \> $(p_1,p_2)\oversetmay{i}p_{12\Phi}$ \> if \> $p_1\oversetmay{i}p_{1\Phi}$ and $p_2\oversetmay{i}$ \kill
        (FMay) \> $(p_1,p_2)\oversetmay{i}p_{12\Phi}$ \> if \> $p_1\oversetmay{i}p_{1\Phi}$ or $p_2\oversetmay{i} p_{2\Phi}$ \\
	\end{tabbing}
\end{definition}
Rules~\emph{(OMust1)}, \emph{(OMust2)}, \emph{(IMust1)}, and~\emph{(IMust2)} 
are concerned with actions being mandatory in one IR-MIA and not being forbidden in the other IR-MIA thus
resulting in must-transitions in the conjunction.
Rules~\emph{(OMay)} and~\emph{(IMay)} introduce may-transitions for actions being allowed in both IR-MIA.
Rules~\emph{(DMay1)} and~\emph{(DMay2)} generate may- and must-transitions 
for inputs being either optional or mandatory in one IR-MIA, and unspecified in the other IR-MIA.
Rules~\emph{(FMust)} and~\emph{(FMay)} apply in all cases where an input is forbidden 
in one IR-MIA and either forbidden, or optional in the other IR-MIA.
Forbidden inputs require additional rules in order to target the new 
failure state $p_{12\Phi}$ instead of $(p_{1\Phi},p_{2\Phi})$.
Note, that $\tau$-steps being present in only one IR-MIA are covered by the rules~\emph{(OMust1)} 
and~\emph{(OMust2)} as $\omega\in O\cup\{\tau\}$ and $q\Oversetgamma{\hat\omega}q'$ 
also includes $q\Oversetgamma{\epsilon}q'$ (\ie{} an empty step such that $q=q'$).

Based on the conjunctive product, we finally obtain the conjunction by (2) 
removing erroneous states $(p_1,p_2)$ (\ie{} states where $p_1$ or $p_2$ 
have mandatory behavior being forbidden in the other state as well as
states requiring a disjunctive transition in the conjunction.
In addition, all states $(p_1',p_2')$ from which  
erroneous states are reachable are also removed (pruned) from $P_1\land P_2$.
\begin{definition}[IR-MIA Conjunction]\label{def:conjunction}
	Given a conjunctive product $P_1\&P_2$, the set $F\subseteq P_1\times P_2$ of \emph{inconsistent states} is defined as the least set satisfying the following rules:
	\begin{tabbing}
		(F7)\quad \= $p_1\oversetmay{a}p_1'$, $p_1\oversetmay{a}p_1''$ with $p_1'\neq p_1''$ and $p_2\oversetmust{a}p_2'$\quad \= implies\quad \= \kill
		(F1) \> $p_1\oversetmust{o}$ and $p_2\not\Oversetmay{o}$ \> implies \> $(p_1,p_2)\in F$ \\
		(F2) \> $p_1\not\Oversetmay{o}$ and $p_2\oversetmust{o}$ \> implies \> $(p_1,p_2)\in F$ \\
		(F3) \> $p_1\oversetmust{i}p_1'\neq p_{1\Phi}$ and $p_2\oversetmust{i}p_{2\Phi}$ \> implies \> $(p_1,p_2)\in F$ \\
		(F4) \> $p_1\oversetmust{i}p_{1\Phi}$ and $p_2\oversetmust{i}p_2'\neq p_{2\Phi}$ \> implies \> $(p_1,p_2)\in F$ \\
		(F5) \> $p_1\oversetmay{a}p_1'$, $p_1\oversetmay{a}p_1''$ with $p_1'\neq p_1''$ and $p_2\oversetmust{a}p_2'$ \> implies \> $(p_1,p_2)\in F$ \\
		(F6) \> $p_2\oversetmay{a}p_2'$, $p_1\oversetmay{a}p_2''$ with $p_2'\neq p_2''$ and $p_1\oversetmust{a}p_1'$ \> implies \> $(p_1,p_2)\in F$ \\
		(F7) \> $(p_1,p_2)\oversetmust{\alpha}r$ and $r\in F$ \> implies \> $(p_1,p_2)\in F$ \\
	\end{tabbing}
	The \emph{conjunction} $P_1\land P_2$ is obtained from $P_1\&P_2$ by deleting all states $(p_1,p_2)\in F$ (and respective transitions).
	If $(p_1,p_2)\in P_1\land P_2$, then we write $p_1\land p_2$, and conjunction $p_1\land p_2$ is defined.
\end{definition}
Rules~\emph{(F1)} to~\emph{(F4)} handle cases
in which an action is mandatory in one IR-MIA and forbidden in the other IR-MIA.
Rules~\emph{(F5)} and~\emph{(F6)} handle cases which would yield
\emph{disjunctive transitions} in the original MIA, but which are not allowed
in our IR-MIA theory, as described before.
Finally, rule~\emph{(F7)} (recursively) removes all states from which
erroneous states are reachable. 
Hence, if there exists a path leading from the initial state 
to an erroneous state, the initial state itself is also removed such that the whole 
conjunction is undefined for this pair of IR-MIA.
Consider the example in Figure~\ref{fig:conjunction}.
First, the conjunctive product is built (\cf{}~Figure~\ref{fig:conjunctive-product-q-and-r}).
The optional inputs \emph{1\euro{}} and \emph{1\pounds{}} of $Q$ become 
mandatory (forbidden) as in $Q$ a refinement to optional, mandatory and 
forbidden inputs are allowed, but $R$ only allows for mandatory input \emph{1\euro{}} and forbidden input \emph{1\pounds{}}.
The inputs for \emph{coffee} and \emph{tea} are only present in one IR-MIA.
Hence, the demonic state is used in this case (\eg{} in $Q\land R$ it holds that $u=u'\land r_d$ with $r_d$ being the demonic state of $R$).
We obtain the conjunction $Q\land R$ from $Q\&R$ by pruning all erroneous states.
Here, $v$ is erroneous as the output of \emph{water} is implicitly forbidden in $Q$, but mandatory in $R$.
Therefore, state $v$ and all its incoming and outgoing transitions are removed.

\begin{figure}[tp]
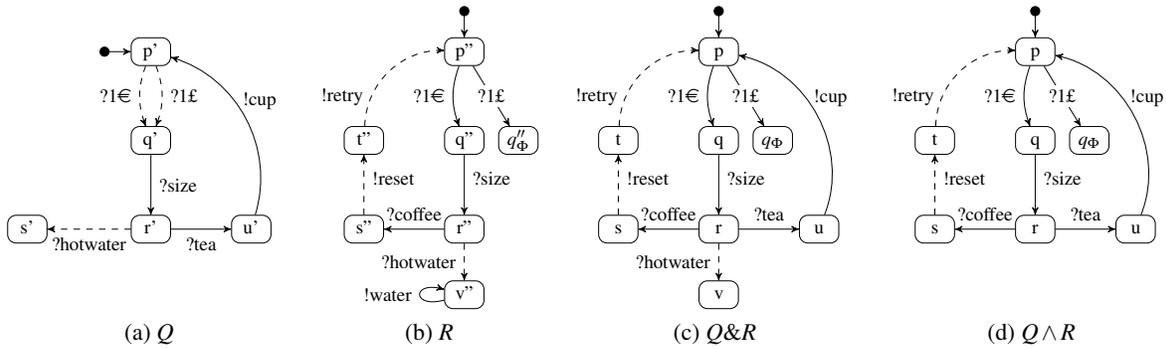

	\hfill
	\subfloat[$Q$]{\label{fig:conjunction-q}\input{figures/conjunction-q.tex}}
	\hfill
	\subfloat[$R$]{\label{fig:conjunction-r}\input{figures/conjunction-r.tex}}
	\hfill
	\subfloat[$Q\&R$]{\label{fig:conjunctive-product-q-and-r}\input{figures/conjunctive-product-q-and-r.tex}}
	\hfill
	\subfloat[$Q\land R$]{\label{fig:conjunction-q-and-r}\input{figures/conjunction-q-and-r.tex}}
	\hfill\strut
	\caption{Example for Conjunction.}\label{fig:conjunction}
\end{figure}

Due to the restrictions on compatibility, the conjunction
of two IR-MIA may be undefined.
However, if the conjunction is defined, then it always results in a proper 
IR-MIA with the expected property with respect to IR-MIA refinement.

\newcounter{and-is-and-counter}
\setcounter{and-is-and-counter}{\value{theorem}}
\begin{theorem}[$\land$ is And]\label{theorem:and-is-and}
	Let $p$ and $q$ be \miaforbidden{} with common alphabets such that $p\land q$ is defined.
	Then, (1) $(\exists r:r\miarefforbidden p\text{ and } r\miarefforbidden q)$ and (2) $r\miarefforbidden p$ and $r\miarefforbidden q$ iff $r\miarefforbidden p\land q$.
\end{theorem}

Furthermore, IR-MIA conjunction is associative.

\newcounter{conjunction-associativity-counter}
\setcounter{conjunction-associativity-counter}{\value{lemma}}
\begin{lemma}[Associativity of IR-MIA Conjunction]\label{lemma:conjunction-associativity}
	Let $P$, $Q$, $R$ be IR-MIA.
	Then, (1) $P\land (Q\land R)$ is defined iff $(P\land Q)\land R$ is defined, and (2) if $P\land (Q\land R)$ is defined, then $S\miarefforbidden P\land (Q\land R)$ iff $S\miarefforbidden (P\land Q)\land R$.
\end{lemma}

In addition, we conclude 
the following compositionality result for \textbf{modal-irioco} with respect to conjunction.

\newcounter{conjunction-counter}
\setcounter{conjunction-counter}{\value{theorem}}
\begin{theorem}[Compositionality of Conjunction of \textbf{modal-irioco}]\label{theorem:conjunction}
	Let $s$, $s'$ and $i$ be \miaforbidden.
	Then it holds that $\left(i\mathmiocophi s\land i\mathmiocophi s'\right)\Rightarrow i\mathmiocophi s\land s'$.
\end{theorem}

\subsection{Disjunction}\label{subsec:disjunction}
Besides conjunction, the original MIA theory also provides an operator for disjunction.
The binary \emph{disjunction} of two MIA, $P$ and $Q$, constructs a MIA $P\lor Q$
integrating all variants permitted by $P$ \emph{or} $Q$ (\ie{} $P\sqsubseteq P\lor Q$ and $Q\sqsubseteq P\lor Q$).
To this end, a fresh initial state is introduced in $P\lor Q$, from which
a non-deterministic choice either leads to the original initial state of $P$ or $Q$.
Intuitively, for this construction to yield a correct result, the initial choice must be
defined as disjunctive transition as described in the previous section.
\begin{definition}[Disjunction for IR-MIA with Disjunctive Transitions]\label{def:disjunction}
	Consider two \miaforbidden $(P_1,I,O,\linebreak\mustarrowsub{P1},\mayarrowsub{P1},p_{1\Phi})$ and $(P,I,O,\mustarrowsub{P2},\mayarrowsub{P2},p_{2\Phi})$ with common alphabets and initial states $p_{10}$ and $p_{20}$.
	Assuming $P_1\cap P_2=\emptyset$, the \emph{disjunction} is defined as $(P_1\cup P_2,I,O,\mustarrow,\mayarrow,p_\Phi)$, where $\mustarrow$ and $\mayarrow$ are the least sets satisfying the conditions $\mustarrowsub{P1}\subseteq\mustarrow$, $\mayarrowsub{P1}\subseteq\mayarrow$, $\mustarrowsub{P2}\subseteq\mustarrow$, $\mayarrowsub{P2}\subseteq\mayarrow$, and the following rules:
	\begin{tabbing}
		(IMay2)\; \= $p_{10}\lor p_{20}\oversetmay{\tau}p_{10}$, $p_{10}\lor p_{20}\oversetmay{\tau}p_{20}$\; \= if\; \= \kill
		(Must) \> $p_{10}\lor p_{20}\oversetmust{\tau}\{p_{10},q_{10}\}$ \> \> \\
		(May) \> $p_{10}\lor p_{20}\oversetmay{\tau}p_{10}$, $p_{10}\lor p_{20}\oversetmay{\tau}p_{20}$ \> \> \\
		(IMust) \> $p_{10}\lor p_{20}\oversetmust{i}P_1'\cup P_2'$ \> if \> $p_{10}\oversetmust{i}P_1'$ and $p_{20}\oversetmust{i}P_2'$ \\
		(IMay1) \> $p_{10}\lor p_{20}\oversetmay{i}p_1'$ \> if \> $p_{10}\oversetmay{i}p_1'$ \\
		(IMay2) \> $p_{10}\lor p_{20}\oversetmay{i}p_2'$ \> if \> $p_{20}\oversetmay{i}p_2'$ \\
	\end{tabbing}
	Furthermore, for each input may-transition to $p_{1\Phi}$ or $p_{2\Phi}$, the target is replaced by $p_{1\Phi}\lor p_{2\Phi}$.
\end{definition}
Rules~\emph{(Must)} and~\emph{(May)} introduce $\tau$-steps for the non-deterministic choice 
following the new initial state.
Rules~\emph{(IMust)}, \emph{(IMay1)} and~\emph{(IMay2)} define the input behaviors of the 
original initial states also for the new initial state, as parallel composition 
does not allow for $\tau$-steps preceding to inputs of 
common actions (\cf{}~Sect.~\ref{sec:compositionality}).

As already demonstrated in the previous section, disjunctive transitions 
obstruct compatibility of \miocophi{} and refinement (\cf~Lemma~\ref{lemma:mioco-disjunctive-transitions}).
However, in contrast to conjunction where disjunctive transitions are only
used for handling corner cases, their role in constructing the disjunction is essential.
Hence, we do not consider disjunction of IR-MIA in the \miocophi{} testing theory in the following.
%
\section{Related Work}\label{sec:related}

We discuss related work on modal conformance relations, 
testing equivalences, alternative formulations of, and extensions to I/O conformance testing 
and composition/decomposition results in I/O conformance testing.

Various interfaces theories have been presented
defining \emph{modal conformance relations} by means
of different kinds of modal refinement relations~\cite{Raclet2011}.
Amongst others, Bauer \etal use interface automata for
compositional reasoning~\cite{Bauer2010}, whereas Alur
\etal characterize modal conformance as alternating
simulation relation on interface automata~\cite{Alur1998}, and
Larsen \etal have shown that 
both views on modal conformance coincide~\cite{Larsen2007}.
Based on our own previous work on modal I/O conformance testing~\cite{Lochau2014,Luthmann2015b},
we present, to the best of our knowledge,
the first comprehensive testing theory by means of a 
modal I/O conformance relation.
More recently, Bujtor \etal proposed testing 
relations on modal transition systems~\cite{Bujtor2015b}
based on (existing) test-suites, rather than being
specification-based as our approach.

In contrast to I/O conformance relations, 
\emph{testing equivalences} constitute a special class 
of (observational) equivalence relations~\cite{Nicola1987,Rensink2007}.
One major difference to \textbf{ioco}-like theories is that actions are 
usually undirected, thus no distinction 
between (input) refusals and (output) quiescence is made as 
in our approach~\cite{Phillips1987,Bourdonov2006}.

Concerning \emph{alterations of and extensions
to I/O conformance testing}, 
Veanes \etal and Gregorio-Rodr{\'i}guez \etal
propose to reformulate
I/O conformance from suspension-trace inclusion 
to an alternating simulation to obtain a more 
fine-grained conformance notion 
constituting a preorder~\cite{Veanes2012,Gregorio2013}.
However, these approaches neither
distinguish optional from mandatory behaviors, nor 
underspecified from forbidden inputs as in our approach.
Heerink and Tretmans extended \ioco
by introducing so-called channels (\ie subsets of I/O labels) 
for weakening the requirement of input-enabledness 
of implementations under test in order to also support refusal testing~\cite{Heerink1997}.
However, their notion of input refusals refers to a global property
rather than being specific to particular states 
and they also do not distinguish mandatory from optional behaviors.
Beohar and Mousavi extend \ioco by replacing IOLTS
with so-called {\em Featured Transition Systems} 
(FTS) and thereby enhance \ioco to express fine-grained behavioral variability as apparent in
software product lines~\cite{Beohar2014}.
As in our approach, FTS allow the environment to explicitly
influence the presence or absence of
particular transitions, whereas compositionality
properties are not considered.

Concerning \emph{(de-)compositionality in I/O conformance testing},
van der Bijl \etal present a compositional version
of \ioco with respect to synchronous parallel composition
on IOTS~\cite{Bijl2004}, whereas Noroozi \etal
consider asynchronously interacting components~\cite{Noroozi2011}.
To overcome the inherent limitations of compositional
I/O conformance testing, Daca \etal introduce alternative criteria for obtaining compositional 
specifications~\cite{Daca2014}.
Concerning decomposition in I/O conformance testing,
Noroozi \etal describe a framework for decomposition of \ioco
testing similar to our setting.
However, all these related approaches neither
distinguish mandatory from optional behaviors, nor support
input refusals as in our approach.

Finally, operators specifically tailored to 
on modal interface specifications, like conjunction and disjunction, have already
been investigated before~\cite{Raclet2009}, but not in the context
of (modal) input/output conformance testing as done in our work.

\section{Conclusion}\label{sec:conclusion}

We proposed a novel foundation for modal I/O-conformance testing theory
based on a modified version of Modal Interface Automata with Input Refusals
and show correctness and (de-)compositionality properties of the
corresponding modal I/O conformance relation called \textbf{modal-irioco}.
As a future work, we are interested in properties of \textbf{modal-irioco}
regarding compositionality with respect to
further operators on IR-MIA, such as interface conjunction~\cite{Luettgen2014} and asynchronous parallel composition~\cite{Noroozi2011}.
Furthermore, we aim at generating test suites exploiting the capabilities of \textbf{modal-irioco}, \ie test cases distinguishing optional from mandatory behaviors, as well as recognizing refused inputs.

\bibliographystyle{eptcs}
\bibliography{content/sources}

\newpage

\begin{appendix}

\section{Overview on Composition Operators}

\begin{table}[h]
\centering
\caption{Overview on Properties of the different Composition Operators on IR-MIA.}\label{table:restrictions-overview}
\begin{tabular}{C{.09\textwidth}|L{.43\textwidth}|L{.36\textwidth}}
	\toprule

	\multicolumn{1}{c}{Operator} & \multicolumn{1}{c}{Defined if} & \multicolumn{1}{c}{Compositional if} \\

	\midrule

	$P\miarefforbidden Q$ & (no restrictions) & implementation is weak must-input-enabled \\\midrule

	$P\land Q$ & 
  (1) actions being mandatory in $P$ are not forbidden in $Q$ (and vice versa),
  (2) no disjunctive transitions occur & $P\land Q$ is defined \\\midrule

	$P\lor Q$ & (never due to disjunctive transitions) & (not applicable) \\\midrule

	$P\parallelforbidden Q$ & $P$ and $Q$ are compatible & implementations are strong must-input-enabled \\\midrule

	$P\hidingforbidden Q$ & $P$ and $Q$ are compatible & specifications and implementations are strong must-input-enabled \\\midrule

	$P\sslashforbidden Q$ & 
  (1) for every transition in $P$ labeled with a shared action, there exist a corresponding transition in $Q$,
  (2) forbidden actions in $P$ are forbidden in $Q$. 
  & 
 (1) all outputs of the composed system are mandatory,
 (2) implementations are weak must-input-enabled
    \\

	\bottomrule
\end{tabular}
\end{table}

\section{Proofs}\label{sec:proofs}

\newcounter{theoremnumbering}

\subsection{Proof of Lemma~\ref{lemma:may-input-enabledness-preservation}}\label{subsec:input-enabledness-preservation-proof}

\setcounter{theoremnumbering}{\value{lemma}}
\setcounter{lemma}{\value{may-input-enabledness-preservation-counter}}

\begin{lemma}
If \miaforbidden $i$ is strong may-input-enabled
then $i'\miarefforbidden i$ is strong may-input-enabled.
\end{lemma}

\setcounter{lemma}{\value{theoremnumbering}}

\begin{proof}
	We prove that \miaforbidden refinement preserves strong may-input-enabledness.
	Assume two \miaforbidden, $P$ and $Q$, with $P\miarefforbidden Q$ and $Q$ being strong may-input-enabled.
	There are two possible reasons why strong may-input-enabledness may be lost under refinement:
	(1) An input transition may be removed, and (2) the target of a transition may be changed to a new state not being may-input-enabled.
	However, under \miaforbidden refinement, both cases are not possible.
	\begin{enumerate}
		\item According to property~5 of Def.~\ref{def:mia-refinement-forbidden}, it must hold that $q\oversetmay{i}q'$ implies $\exists p'.p\oversetmay{i}\Oversetmay{\epsilon}p'$ and $(p',q')\in\mathcal{R}$.
		But, it is impossible to remove input transitions under \miaforbidden refinement.
		\item Now we have to look at the possibility of changing the target of an input transition to a new state not being strongly may-input-enabled.
		Properties~3 and~6 ensure that output transitions must be preserved if they are mandatory or may be removed if they are optional.
		However, \miaforbidden refinement only allows to change the target to a new state, if the behavior of that new state is equivalent
		to the old (may-input-enabled) target state, thus also being strong may-input-enabled.
		Otherwise $(p',q')\in\mathcal{R}$ would be violated.
		The same holds for input transitions, but with one exception: input transitions may change their target to the failure state
		under \miaforbidden refinement.
		By definition, the failure state does not have any
		(input) behavior, but this does not obstruct input-enabledness
		as the failure state is excluded from this requirement.
	\end{enumerate}
	Under \miaforbidden refinement, it is impossible to remove input transitions or change
	the target of a transition to a new state not being strong may-input-enabled.
	Therefore, strong may-input-enabledness is always preserved under IR-MIA refinement.
\end{proof}

\subsection{Proof of Theorem~\ref{theorem:mioco-refinement-compatability}}\label{subsec:mioco-ref-compat-proof}

\setcounter{theoremnumbering}{\value{theorem}}
\setcounter{theorem}{\value{mioco-refinement-compatability-counter}}

\begin{theorem}
	Let $i,s$ be \miaforbidden, $i$ being weak may-input-enabled
	and $i\mathmiocophi s$.
	Then for each $i'\miarefforbidden i$ there exists
	$s'\miarefforbidden s$ such that $i'\mathmiocophi s'$ holds.
\end{theorem}

\setcounter{theorem}{\value{theoremnumbering}}

\begin{proof}
	We construct a \emph{unifying specification} $s_u$ serving as $s'$
	for all $i'\miarefforbidden i$, by initially setting $s_u=s$.
	As $i$ is may-input-enabled, inputs in $i'$ are either
	may-failure, must-failure, or implemented as mandatory behavior.
	Hence, we do not have to modify $s_u$ as optional inputs of
	$i$ are either also optional in $s_u$ or unspecified
	thus allowing input behaviors to be may-failure, must-failure, optional as well as mandatory.
	However, it is possible that $\delta\in\mayout(i'\mayafter\sigma)$
	although $\delta\notin\mayout(s\mayafter\sigma)$ if there are states
	having only optional outputs.
	In this case, we add a $\tau$-transition to every state in $s_u$, having only
	optional output behavior leading to a fresh must-quiescent
	state without any output transitions, such that $\delta\in\mayout(s_u\mayafter\sigma)$.
	From Def.~\ref{def:mia-refinement-forbidden}, it follows that $s_u\miarefforbidden s$.
	Thus, $\forall i'\miarefforbidden i: i'\mathmiocophi s_u$ holds and therefore the claim holds.
\end{proof}

\subsection{Proof of Theorem~\ref{theorem:soundness-completeness-mioco}}\label{subsec:soundness-completeness-proof}

\setcounter{theoremnumbering}{\value{theorem}}
\setcounter{theorem}{\value{soundness-completeness-mioco-counter}}

\begin{theorem}[\textbf{modal-irioco} is correct]
	Let $i,s$ be \miaforbidden, $i$ being weak may-input-enabled.
	\begin{enumerate}
		\item If $i\mathmiocophi s$, then for all $i'\miarefforbidden i$,
		there exists $s'\miarefforbidden s$ such that $i'\mathioco s'$.
		\item If there exists $i'\miarefforbidden i$
		such that $i'\mathioco s'$ does not hold for any $s'\miarefforbidden s$,
		then $i\mathmiocophi s$ does not hold.
	\end{enumerate}
\end{theorem}

\setcounter{theorem}{\value{theoremnumbering}}

\begin{proof}
	We prove both parts separately.
	\begin{itemize}
		\item We make use of the \emph{unifying specification} $s_u$ from the proof of Theorem~\ref{theorem:mioco-refinement-compatability} (\cf Sect.~\ref{subsec:mioco-ref-compat-proof}).
			Therefore, it holds that $\forall i'\miarefforbidden i: i'\mathioco s_u$.

		\item For this part, we rely on the \emph{unifying specification} $s_u$.
			Let $i'\miarefforbidden i$ be a \miaforbidden such that
			$i'\mathioco s_u$ does not hold, \ie there exists a trace
			$\sigma\in\maystraces(s_u)$ such that $\mayout(i'\mayafter\sigma)\nsubseteq\mayout(s_u\mayafter\sigma)$.
			Thus, we have $\mayout(i\mayafter\sigma)\neq\emptyset$.
			From the construction of $s_u$, it follows
			that there is an $\omega\in\mayout(i\mayafter\sigma)\setminus\mayout(s_u\mayafter\sigma)$.
			But, then $i\mathmiocophi s$ does not hold since $\mayout(i\mayafter\sigma)\nsubseteq\mayout(s\mayafter\sigma)$.
	\end{itemize}
	Hence, \miocophi is sound and complete.
\end{proof}

\subsection{Proof of Theorem~\ref{theorem:mioco-preorder}}\label{subsec:mioco-preorder-proof}

\setcounter{theoremnumbering}{\value{theorem}}
\setcounter{theorem}{\value{mioco-preorder-counter}}

\begin{theorem}
	\miocophi is a preorder on the set of weak may-input-enabled \miaforbidden.
\end{theorem}

\setcounter{theorem}{\value{theoremnumbering}}

\begin{proof}
	Let $p, q, r$ be \miaforbidden such that $p$ and $q$ are weak may-input-enabled and $p\mathmiocophi q$ and $q\mathmiocophi r$.
	It holds by Def.~\ref{def:mioco} that $p\mathmiocophi p$, \ie $\mathmiocophi$ is reflexive.
	It remains to be shown that $p\mathmiocophi r$, \ie (a) for all $\sigma\in\maystraces(r)$, $\mayout(p \mayafter \sigma) \subseteq \mayout(r \mayafter \sigma)$ and
	(b) for all $\maystraces(p)$,\linebreak$\mustout(r \mayafter \sigma) \subseteq \mustout(p \mayafter \sigma)$.
	Let $\sigma\in\maystraces(r)$. If $\sigma\in\maystraces(q)$, then (a) and (b) follow from transitivity of $\subseteq$.

	The case of $\sigma\notin\maystraces(q)$ remains.

	Suppose (a) fails for a $\sigma\in\maystraces(r)\setminus\maystraces(q)$, \ie such a $\sigma$ exists.
	Trace $\sigma$ decomposes into $\sigma_1 \cdot a \cdot \sigma_2$ where $\sigma_1\in\maystraces(q)$ but $\sigma_1\cdot a \notin\maystraces(q)$.
	Since $\mayout(p \mayafter \sigma_1) \subseteq \mayout(q \mayafter \sigma_1)$, $a\notin O\cup\{\delta,\varphi\}$.
	Otherwise, $a\in I$ contradicts weak may-input-enabledness of $q$.
	Thus, $\sigma\in\maystraces(q)$.

	Case (b) remains for $\sigma\in\maystraces(p)\setminus\maystraces(q)$.
	We show that such a $\sigma$ again contradicts the assumptions of the theorem.
	As $\sigma\notin\maystraces(q)$, $\sigma$ decomposes into a prefix $\sigma_1\in\maystraces(q)$ and a postfix $a\cdot\sigma_2$ such that $\sigma_1\cdot a\notin\maystraces(q)$.
	Since $\sigma_1\cdot a\in\maystraces(p)$ and $p\mathmiocophi q$, $a\notin O\cup\{\delta,\varphi\}$.
	Hence $a\in I$, but as stated above, this contradicts the assumption that $q$ is weak input-enabled.

	From reflexivity and transitivity of $\mathmiocophi$ it follows that $\mathmiocophi$
	is indeed a preorder on weak may-input-enabled \miaforbidden.
\end{proof}

\subsection{Proof of Theorem~\ref{theorem:mioco-demonic-completion}}\label{subsec:mioco-completion-proof}

\setcounter{theoremnumbering}{\value{theorem}}
\setcounter{theorem}{\value{mioco-demonic-completion-counter}}

\begin{theorem}
	Let $i$, $s$ be \miaforbidden with $i$ being weak must-input-enabled.
	Then $i\mathmiocophi \Xi(s)$ if $i\mathmiocophi s$.
\end{theorem}

\setcounter{theorem}{\value{theoremnumbering}}

\begin{proof}
	Let $i,s$ be \miaforbidden with $i$ being weak must-input-enabled.
	We prove that it holds that (1) $\forall\sigma\in\maystraces(s):\mayout(i\mayafter\sigma)\subseteq\mayout(s\mayafter\sigma)\Rightarrow\forall\sigma\in\maystraces(\Xi(s)):\mayout(i\mayafter\sigma)\subseteq\linebreak\mayout(\Xi(s)\mayafter\sigma)$ and (2) $\forall\sigma\in\maystraces(i):\mustout(s\mayafter\sigma)\subseteq\mustout(i\mayafter\sigma)\Rightarrow\forall\sigma\in\maystraces(i):\mustout(\Xi(s)\mayafter\sigma)\subseteq\mustout(i\mayafter\sigma)$.
	\begin{enumerate}
		\item Because of $i\mathmiocophi s$, the subset relation holds for all \maystraces specified by $s$, \ie $\forall\sigma\in\linebreak\maystraces(s):\mayout(i\mayafter\sigma)\subseteq\mayout(\Xi(s)\mayafter\sigma)$.
		Therefore, we have to prove the assumption for all $\maystraces(\Xi(s))\setminus\maystraces(s)$, \ie all traces not specified by $s$.
		Let $\sigma=\sigma'\cdot i\cdot\sigma''$ with $\sigma'\in\maystraces(s)$, $i\in I$ being an unspecified input such that $\sigma\notin\maystraces(s)$, and $\sigma\in\maystraces(\Xi(s))$.
		We prove that for all $\sigma$ it holds that $\mayout(\Xi(s)\mayafter\sigma)=O\cup\{\varphi,\delta\}$ because $\mayout(i\mayafter\sigma)\subseteq O\cup\{\varphi,\delta\}$ is always true.
		$O\subseteq\mayout(\Xi(s)\mayafter\sigma)$ because for every \miadc it holds by definition that $\{(q_\Omega,\lambda,q_\chi)\mid \lambda\in O\}\subseteq\mayarrow$ and $\{(q_\chi,\tau,q_\Omega)\}\subseteq\mayarrow$.
		$\varphi\in\mayout(\Xi(s)\mayafter\sigma)$ because $\{(q_\Omega,\lambda,q_\chi)\mid \lambda\in I\}\subseteq\mayarrow$ and $\mustarrow\cap\{(q_\Omega,\lambda,q_\chi)\mid \lambda\in I\}=\emptyset$, \ie there are optional input transitions in $q_\Omega$.
		$\delta\in\mayout(\Xi(s)\mayafter\sigma)$ because $\forall q\in Q:\mustarrow\cap(\{(q_\chi,\lambda,q)\mid \lambda\in O\}\cup\{(q_\Omega,\lambda,q)\mid \lambda\in O\})=\emptyset$.

		\item Because of $i\mathmiocophi s$, the subset relation holds for all \maystraces specified by $s$, \ie $\forall\sigma\in\linebreak\maystraces(i):\mustout(\Xi(s)\mayafter\sigma)\subseteq\mayout(i\mayafter\sigma)$.
		Therefore, the assumption remains to be proven for all $\sigma\in\maystraces(\Xi(s))\setminus\maystraces(s)$, \ie all traces not specified by $s$.
		Let $\sigma=\sigma'\cdot i\cdot\sigma''$ with $\sigma'\in\maystraces(s)$, $i\in I$ being an unspecified input such that $\sigma\notin\maystraces(s)$, and $\sigma\in\maystraces(\Xi(s))$.
		$\forall o\in O:o\notin\mustout(\Xi(s)\mayafter\sigma)$ because $\forall q\in Q:\mustarrow\cap(\{(q_\chi,\lambda,q)\mid \lambda\in O\}\cup\{(q_\Omega,\lambda,q)\mid \lambda\in O\})=\emptyset$.
		$\varphi\notin\mustout(\Xi(s)\mayafter\sigma)$ because $\{(q_\chi,i,q_\Phi),(q_\Omega,i,q_\Phi)\}\cap\mustarrow=\emptyset$.
		$\delta\notin\mustout(\Xi(s)\mayafter\sigma)$ because $\mayarrow\cap\{(q_\Omega,\lambda,q_\chi)\mid \lambda\in O\}\neq\emptyset$.
		Thus, for all $\sigma$ it holds that $\mustout(\Xi(s)\mayafter\sigma)=\emptyset$.
	\end{enumerate}
	Therefore, both assumptions hold and $i\mathmiocophi s\Rightarrow i\mathmiocophi \Xi(s)$.
\end{proof}

\subsection{Proof of Theorem~\ref{theorem:compositionality-multicast}}\label{subsec:comp-multicast-proof}

\setcounter{theoremnumbering}{\value{theorem}}
\setcounter{theorem}{\value{compositionality-multicast-counter}}

\begin{theorem}[Compositionality of \textbf{modal-irioco}]
	Let $s_1$, $s_2$, $i_1$, and $i_2$ be \miaforbidden with $i_1$ and $i_2$ being strong must-input-enabled, and $s_1$ and $s_2$ being compatible.
	Then it holds that
	$\left(i_1\mathmiocophi s_1\land i_2\mathmiocophi s_2\right)\Rightarrow i_1\parallelforbidden i_2\mathmiocophi s_1\parallelforbidden s_2$.
\end{theorem}

\setcounter{theorem}{\value{theoremnumbering}}

\begin{proof}
	Let $s_1$, $s_2$, $i_1$, and $i_2$ be \miaforbidden with $i_1$ and $i_2$ being must-input-enabled, and $s_1$ and $s_2$ being compatible.
	Additionally, $i_1\mathmiocophi s_1$ and $i_2\mathmiocophi s_2$ hold.
	In order to prove $i_1\parallelforbidden i_2\mathmiocophi s_1\parallelforbidden s_2$, we prove that (1) $\forall\sigma\in\maystraces(s_1\parallelforbidden s_2):\mayout(i_1\parallelforbidden i_2\mayafter\sigma)\subseteq\mayout(s_1\parallelforbidden s_2\mayafter\sigma)$ and (2) $\forall\sigma\in\maystraces(s_1\parallelforbidden s_2):\mustout(s_1\parallelforbidden s_2\mayafter\sigma)\subseteq\mustout(i_1\parallelforbidden i_2\mayafter\sigma)$.
	\begin{enumerate}
		\item Let $\omega\in\mayout(i_1\parallelforbidden i_2\mayafter\sigma)$ such that $\omega\in\mayout(i_1\mayafter\sigma)$.
		We prove that $\omega\in\mayout(s_1\parallelforbidden s_2\mayafter\sigma)$ holds.
		$\mayout(s_1\parallelforbidden s_2\mayafter\sigma)\neq\emptyset$ because $\sigma\in\maystraces(s_1\parallelforbidden s_2)$.
		We now have to distinguish between $\omega\in O$ and $\omega=\delta$.
		If $\mayout(i_1\parallelforbidden i_2\mayafter\sigma)\subseteq O$, then $\omega\in\mayout(s_1\parallelforbidden s_2\mayafter\sigma)$ because otherwise $i_1$ would have more output behavior than $s_1$ such that $i_1\mathmiocophi s_1$ would not hold.
		If $\mayout(i_1\parallelforbidden i_2\mayafter\sigma)=\{\delta\}$, then $\omega\in\mayout(s_1\parallelforbidden s_2\mayafter\sigma)$ because otherwise $i_1$ would have less mandatory behavior than $s_1$ such that $i_1\mathmiocophi s_1$ would not hold.

		Additionally, we have to consider pruning applied in $s_1\parallelforbidden s_2$ but not in $i_1\parallelforbidden i_2$.
		In order for pruning to occur in $s_1\parallelforbidden s_2$ but not in $i_1\parallelforbidden i_2$, there must be an optional input $i\in A_1\cap A_2$ of $s_1$ becoming mandatory in $i_1$.
		Then in $s_1\parallelforbidden s_2$, pruning occurs as it is a new error (provided that $s_2$ performs a matching output).
		In $i_1\parallelforbidden i_2$, there is no pruning as $p_1\not\oversetmust{i}$ does not hold.
		However, due to the definition of illegal states and pruning (\cf Def.~\ref{def:illegal-states} and Def.~\ref{def:parallel-composition}), all states being able to reach an illegal state through outputs and all their incoming and outgoing transitions are removed.
		Therefore, the least removed action of a trace is an input.
		Hence, that state in $s_1$ which had the removed input as an outgoing transition is underspecified and traces including that input are never checked (only traces $\sigma\in\maystraces(s_1\parallelforbidden s_2)$ are checked).
		If the optional input $i$ remains optional or is removed, then the same pruning is applied in $i_1\parallelforbidden i_2$ because the state which should have the input is an illegal state.
		Therefore, $\forall\sigma\in\maystraces(s_1\parallelforbidden s_2):\mayout(i_1\parallelforbidden i_2\mayafter\sigma)\subseteq\mayout(s_1\parallelforbidden s_2\mayafter\sigma)$ is always true.

		\item Let $\omega\in\mustout(s_1\parallelforbidden s_2\mayafter\sigma)$ and $\omega\in\mustout(s_1\mayafter\sigma)$.
		$\omega\in\mustout(i_1\parallelforbidden i_2\mayafter\sigma)$ must hold because otherwise $i_1$ would have less mandatory output behavior than $s_1$ requires such that $i_1\mathmiocophi s_1$ does not hold.
		Unlike the first part of this proof, the second part does not need to consider pruning because if there is pruning in $s_1\parallelforbidden s_2$, then there is also pruning in $i_1\parallelforbidden i_2$.
		This is because pruning is only performed if the input of a common action is optional, leading to the failure state, or the input is unspecified.
		If an input is optional in $i_1\parallelforbidden i_2$, then it is optional or underspecified in $s_1\parallelforbidden s_2$ (meaning, there is pruning on both sides).
		If an input is leading to the failure state in $i_1\parallelforbidden i_2$, then it is optional or leading to failure state in $s_1\parallelforbidden s_2$ (again, pruning on both sides).
		Underspecification is only possible for $s_1$ or $s_2$ so in this case pruning only takes place in $s_1\parallelforbidden s_2$.
		Therefore, $\forall\sigma\in\maystraces(s_1\parallelforbidden s_2):\mayout(i_1\parallelforbidden i_2\mayafter\sigma)\subseteq\mayout(s_1\parallelforbidden s_2\mayafter\sigma)$ is always true.
	\end{enumerate}
\end{proof}

\subsection{Proof of Lemma~\ref{lemma:parallel-composition-associativity}}\label{subsec:associativity-proof}

To prove associativity of IR-MIA parallel composition, we first define a transformation of MIA according to Bujtor \etal~\cite{Bujtor2015a}, and we prove the transformation to be correct regarding parallel composition.
Then, associativity of IR-MIA parallel composition directly follows, because the parallel composition of MIA according to Bujtor \etal is associative.

\begin{definition}[Transformation of \miaallowed to \miaforbidden]\label{def:mia-transformation}
	The transformation function $\mathcal{T}:\miaallowed\rightarrow\miaforbidden$ is defined as $\mathcal{T}(P):=(P',I_P,O_P,\mustarrow^{P'},\mayarrow^{P'},p_\Phi)$ with:
	\begin{align*}
		P' & = (P\setminus\{u_P\}) \dot{\cup} \{p_\Phi\} \\
		\mustarrow^\Phi & = \left\{(p,i,p_\Phi)\mid p\in P,i\in I, p\not\oversetmay{i}\right\} \\
		\mustarrow^{P'} & = \mustarrow \cup \mustarrow^{\Phi} \\
		\mayarrow^{P'} & = \left( \mayarrow \cap \left( P'\times A_P^\tau\times P'\right)\right) \cup \mustarrow^\Phi
	\end{align*}
	$\dot{\cup}$ denotes the \textup{disjoint union}, \ie it holds that $(P\setminus\{u_P\})\cap\{p_\Phi\}=\emptyset$.
\end{definition}

\begin{theorem}\label{theorem:multicast-composition-correctness}
	Let $P$ and $Q$ be \miaallowed. Then it holds that $\mathcal{T}(P)\parallelforbidden\mathcal{T}(Q)\cong\mathcal{T}(P\parallelallowed Q)$.
\end{theorem}

\begin{proof}
	Let $P$ and $Q$ be \miaallowed, $A=\mathcal{T}(P)\parallelforbidden\mathcal{T}(Q)$, and $B=\mathcal{T}(P\parallelallowed Q)$.
	To prove $A\cong B$, we have to prove that $(S_A,I_A,O_A,\mustarrow^A,\mayarrow^A,a_\Phi)\cong(S_B,I_B,O_B,\mustarrow^B,\mayarrow^B,b_\Phi)$.
	\begin{itemize}
		\item When constructing the parallel product, the sets of states of both automata is equal because the rules for building the parallel product are similar (\cf Def.~4 of Bujtor \etal~\cite{Bujtor2015a} and Def.~\ref{def:parallel-product}).
		The \miaforbidden parallel product uses two additional rules to ensure inputs being implicitly forbidden for \miaallowed are explicitly forbidden in \miaforbidden.
		Next, the sets of new errors of the parallel composition with multicast are equal because \miaforbidden share the two rules of \miaallowed (\cf definitions for new errors in Def.~5 of Bujtor \etal~\cite{Bujtor2015a} and Def.~\ref{def:parallel-composition}).
		Again, \miaforbidden parallel composition needs two additional rules to ensure explicitly forbidden inputs (which are implicitly forbidden in \miaallowed) are taken into account.
		Additionally, \miaallowed parallel composition incorporates inherited errors not being defined for \miaforbidden parallel composition because they are already explicitly forbidden through \miaforbidden parallel product rules \emph{May4/Must4} and \emph{May5/Must5}.
		As a consequence, the sets of illegal states are equal because the new errors are equivalent (and the inherited errors are taken into account for \miaforbidden).
		In the last step, an equal set of illegal states is pruned in \miaallowed and \miaforbidden.
		Furthermore, the transformations $\mathcal{T}(P)$ and $\mathcal{T}(Q)$ remove the universal states from $P$ and $Q$, respectively, and add the failure state.
		After using the \miaforbidden parallel composition with multicast, exactly one failure state is left.
		When performing $P\parallelallowed Q$, exactly one universal state remains.
		The transformation then removes the universal state and adds a failure state.
		Therefore, it holds that $S_A=S_B$.

		\item The transformation does not change the set of inputs and outputs.
		Therefore, $I_A=(I_P\cup I_Q)\setminus(O_P\cup O_Q)=I_B$ and $O_A=O_P\cup O_Q=O_B$.

		\item The transformation of $P\parallelallowed Q$ adds must-transitions with the failure state as their target, \ie it adds the set $\left\{(b,i,b_\Phi)\mid b\in S_B,i\in I, b\not\oversetmay{i}\right\}$ in order to obtain $B$.
		All these transitions are also contained in $A$ because the transformation of $P$ and $Q$ adds equal transitions to the respective sets of must-transitions.
		$P\parallelforbidden Q$ then combines all failure states into one failure state.
		Additionally, the sets of illegal states of the composition are equal (as described above).
		Therefore, an equal set of must-transitions is pruned, and it holds that $\mustarrow^A=\mustarrow^B$.

		\item The transformation of $P\parallelallowed Q$ adds may-transitions with the failure state as their target, \ie it adds the set $\left\{(b,i,b_\Phi)\mid b\in S_B,i\in I, b\not\oversetmay{i}\right\}$.
		Furthermore, it removes all transitions with the universal state as their target, \ie it removes the set $\mayarrow^B\cap\left(S_B'\times A_B^\tau\times S_B'\right)$ in order to obtain $B$ (with $S_B'=(S_B\setminus\{u_B\}) \dot{\cup} \{b_\Phi\}$).
		All these transitions are also contained in $A$ because the transformation of $P$ and $Q$ adds equal transitions to their sets of may-transitions.
		$P\parallelforbidden Q$ then combines all failure states into one failure state.
		Additionally, the sets of illegal states of the composition are equal (as described above).
		Therefore, an equal set of may-transitions is pruned, and it holds that $\mayarrow^A=\mayarrow^B$.

		\item The transformation removes the universal state and adds a failure state to $P$ and $Q$, respectively.
		The \miaforbidden parallel composition with multicast combines all failure states into one failure state.
		The \miaallowed parallel composition with multicast combines all universal states into one universal state.
		The transformation then removes the universal state and adds a failure state.
		Hence, it holds that $a_\Phi=b_\Phi$.
	\end{itemize}
	Therefore, it holds that $(S_A,I_A,O_A,\mustarrow^A,\mayarrow^A,a_\Phi)\cong(S_B,I_B,O_B,\mustarrow^B,\mayarrow^B,b_\Phi)$, \ie $\mathcal{T}(P)\parallelforbidden\mathcal{T}(Q)\linebreak\cong\mathcal{T}(P\parallelallowed Q)$.
\end{proof}

Now, correctness of Lemma~\ref{lemma:parallel-composition-associativity} directly follows.

\setcounter{theoremnumbering}{\value{lemma}}
\setcounter{lemma}{\value{parallel-composition-associativity-counter}}

\begin{lemma}[Associativity of IR-MIA Parallel Composition]
	Let $P$, $Q$, $R$ be IR-MIA.
	It holds that $(P\parallelforbidden Q)\parallelforbidden R=P\parallelforbidden(Q\parallelforbidden R)$.
\end{lemma}

\setcounter{lemma}{\value{theoremnumbering}}

\begin{proof}
	Because of Theorem~\ref{theorem:multicast-composition-correctness} and associativity of parallel composition according to Bujtor \etal~\cite{Bujtor2015a}, it follows that $(P\parallelforbidden Q)\parallelforbidden R=P\parallelforbidden(Q\parallelforbidden R)$.
\end{proof}

\subsection{Proof of Theorem~\ref{theorem:compositionality-hiding}}\label{subsec:compositionality-hiding-proof}

To prove Theorem~\ref{theorem:compositionality-hiding}, we first transfer the hiding operator of MIA according to Bujtor \etal~\cite{Bujtor2014} to IR-MIA, and prove some intermediate results.
The definition of MIA hiding according to Bujtor \etal is directly transferable to \miaforbidden because hiding only affects output behavior (which is not changed by the transformation from MIA according to Bujtor \etal to \miaforbidden).

\begin{definition}[Hiding for \miaforbidden]\label{def:hiding-forbidden}
	Given a \miaforbidden $P=\left(P,I,O,\mustarrow^P,\mayarrow^P,p_\Phi\right)$ and $L\subseteq O$, then $P\mathhiding L$ is a \miaforbidden $P/L=_\textit{def}\left(P,I,O\setminus L,\mustarrow^{P/L},\mayarrow^{P/L},p_\Phi\right)$, where
	\begin{align*}
		\transitionarrow_\gamma^o & = \left\{\left(p_1,o,p_2\right)\mid p_1,p_2\in P, o\in L,p_1\transrel{o}_\gamma p_2\right\} \\
		\transitionarrow_\gamma^\tau & = \left\{\left(p_1,\tau,p_2\right)\mid p_1,p_2\in P, o\in L,p_1\transrel{o}_\gamma p_2\right\} \\
		\transitionarrow_\gamma^{P/L} & = \left(\transitionarrow_\gamma^P\setminus\transitionarrow_\gamma^o\right)\cup\transitionarrow_\gamma^\tau
	\end{align*}
\end{definition}

As described above, the hiding operation can easily be transfered to \miaforbidden because it only affects output behavior.
In fact, we define hiding in such a way that it commutes with the transformation described in Def.~\ref{def:mia-transformation}.
This is due to the transformation (\cf Def.~\ref{def:mia-transformation}) not affecting any output behavior.

\begin{corollary}
	Let $P$ be a \miaallowed with $O_P$ being the set of output actions of $P$, and a set of actions $L\subseteq O_P$.
	Then $\mathcal{T}(P\mathhiding L)\cong\mathcal{T}(P)\mathhiding L$.
\end{corollary}

Let $\parallelallowed$ and $\hidingallowed$ denote composition with multicast and hiding according to Bujtor \etal~\cite{Bujtor2014}, respectively.
From Bujtor \etal~\cite{Bujtor2014}, it follows that $P\hidingallowed Q=(P\parallelallowed Q)/S$ with $S=A_P\cap A_Q$ holds for \miaallowed.
This means that building the parallel composition with hiding is equal to first building the parallel composition with multicast and, afterward, hide the common actions.
We can show that this also holds for \miaforbidden.
For practical purposes, this means that a tool being able to build the parallel composition with multicast only needs an extension which applies hiding (instead of creating an additional tool).

\begin{lemma}\label{lemma:multicast-hiding-coherence}
	Let $P$ and $Q$ be \miaforbidden and $S=A_P\cap A_Q$ the set of common actions of $P$ and $Q$. Then, $P\hidingforbidden Q=(P\parallelforbidden Q)/S$.
\end{lemma}

Next, we look at hiding in the context of \miocophi.
In general, hiding does not preserve \miocophi, \ie $i\mathmiocophi s\Rightarrow(i\mathhiding L)\mathmiocophi(s\mathhiding L)$ does not hold.
However, if we require the specification to be may-input-enabled, then hiding preserves \miocophi.

\begin{theorem}\label{theorem:mioco-hiding}
	Let $i,s$ be \miaforbidden with $i$ being weakly must-input-enabled and $s$ being weakly may-input-enabled, $O$ be the set of outputs, and $L\subseteq O$.
	Then $i\mathmiocophi s\Rightarrow(i\mathhiding L)\mathmiocophi(s\mathhiding L)$.
\end{theorem}

\begin{proof}
	Let $i$ and $s$ be \miaforbidden with $i\mathmiocophi s$, $i$ being weakly must-input-enabled, and $s$ being weakly may-input-enabled, $O$ be the set of outputs of $i$ and $s$, and $L\subseteq O$.
	Additionally, $s'=s\mathhiding L$, $i'=i\mathhiding L$, $\sigma\in\maystraces(s)$, and $\sigma'$ be the corresponding trace $\sigma'\in\maystraces(s')$ where the hidden actions are removed from $\sigma$.
	To prove Theorem~\ref{theorem:mioco-hiding}, we prove that (1) $\forall\sigma\in\maystraces(s):(\mayout(i\mayafter\sigma)\subseteq\mayout(s\mayafter\sigma))\Rightarrow\forall\sigma'\in\maystraces(s'):(\mayout(i'\mayafter\sigma)\subseteq\mayout(s'\mayafter\sigma))$ and (2) $\forall\sigma\in\linebreak\maystraces(s):(\mustout(s\mayafter\sigma)\subseteq\mustout(i\mayafter\sigma))\Rightarrow\forall\sigma'\in\maystraces(s'):(\mustout(s'\mayafter\sigma)\subseteq\linebreak\mustout(i'\mayafter\sigma))$.
	We only look at $\sigma'$ with $\sigma'\neq\sigma$ because otherwise (1) and (2) directly follow.
	\begin{enumerate}
		\item Assume, (1) does not hold.
		Then, there exists an $\omega\neq\tau$ such that $\omega\in\mayout(i\mayafter\sigma)$, $\omega\in\mayout(s\mayafter\sigma)$, $\omega\in\mayout(i'\mayafter\sigma')$, and $\omega\notin\mayout(s'\mayafter\sigma')$.
		However, with $s$ being may-input-enabled, we impose that $i$ implements every input $i\in I$ for every state $q\in Q$ such that $q_s\oversetmay{i}q_s'\Rightarrow q_i\oversetmay{i}q_i'$, $q_s\oversetmay{i}q_s'\Rightarrow q_i\oversetmust{i}q_i'$, or $q_s\oversetmay{i}q_s'\Rightarrow q_i\oversetmust{i}q_{i\Phi}$.
		This means, we prescribe how $i$ should behave after $\sigma$.
		Therefore, $i'$ cannot have any additional output behavior after $\sigma'$ not being in $s'$ after $\sigma'$.

		\item Assume, (2) does not hold.
		Then there exists an $\omega$ such that $\omega\in\mustout(s\mayafter\sigma)$, $\omega\in\linebreak\mustout(i\mayafter\sigma)$, $\omega\in\mustout(s'\mayafter\sigma')$, and $\omega\notin\mustout(i'\mayafter\sigma')$.
		However, this is not possible because in this case, there would be more output behavior in $i$ than in $s$ such that $i\mathmiocophi s$ would not hold.
	\end{enumerate}
\end{proof}

Now, we prove that parallel composition with hiding preserves strong must-input-enabledness.
For this, we first prove preservation for parallel composition with multicast.

\begin{lemma}\label{lemma:composition-input-enabledness-preservation}
	Let $P$ and $Q$ be strongly must-input-enabled \miaforbidden.
	If it holds that $\forall p\in P:\forall i\in I_P\cap O_Q:p\not\oversetmay{i}p_\Phi$ and $\forall q\in Q:\forall i\in I_Q\cap O_P:q\not\oversetmay{i}q_\Phi$, then $P\parallelforbidden Q$ is must-input-enabled.
\end{lemma}

\begin{proof}
	Let $I_P$ be the set of inputs of $P$, $I_Q$ the set of inputs of $Q$, and $I=(I_P\cup I_Q)\setminus(O_P\cup O_Q)$.
	\miaforbidden parallel product rule \emph{(May1/Must1)} ensures that that all inputs $I_P\setminus(I_Q\cup O_Q)$ are accepted in every state of $P\otimesforbidden Q$.
	Rule \emph{(May2/Must2)} ensures the same for all inputs $I_Q\setminus(I_P\cup O_P)$, and rule \emph{(May3/Must3)} for all inputs $I_P\cap I_Q$.
	Additionally, no inputs are pruned because there are no new errors (and no illegal states) due to the fact that $P$ and $Q$ are strongly must-input-enabled, and there is no reachable input behavior of $I_P\cap O_Q$ or $I_Q\cap O_P$ with the failure state as its target.
\end{proof}

Previously, it has been proven that \miaforbidden parallel composition with multicast preserves strong must-input-enabledness if there are no transitions having a common action with the failure state as their target (\cf Lemma~\ref{lemma:composition-input-enabledness-preservation}).
Because of Lemma~\ref{lemma:multicast-hiding-coherence}, we can transfer that result to parallel composition with hiding.
This is possible as hiding only affects output behavior, whereas input behavior remains unchanged (\ie no input transition becomes internal behavior).

\begin{corollary}\label{corollary:hiding-input-enabledness-preservation}
	Let $P$, $Q$ be strong must-input-enabled \miaforbidden.
	If it holds that $\forall p\in P:\forall i\in I_P\cap O_Q:p\not\oversetmay{i}p_\Phi$ and $\forall q\in Q:\forall i\in I_Q\cap O_P:q\not\oversetmay{i}q_\Phi$, then $P\hidingforbidden Q$ is must-input-enabled.
\end{corollary}

Now, we can prove Theorem~\ref{theorem:compositionality-hiding}.

\setcounter{theoremnumbering}{\value{theorem}}
\setcounter{theorem}{\value{compositionality-hiding-counter}}

\begin{theorem}[Compositionality of \miocophi Regarding Parallel Composition with Hiding]
	Let $s_1$, $s_2$, $i_1$, and $i_2$ be strongly must-input-enabled \miaforbidden.
	Then $\left(i_1\mathmiocophi s_1\land i_2\mathmiocophi s_2\right)\Rightarrow i_1\hidingforbidden i_2\mathmiocophi s_1\hidingforbidden s_2$ if $s_1$ and $s_2$ are compatible, $\forall q\in Q_{s1}:\forall i\in I_{s1}\cap O_{s2}:q\not\oversetmay{i}q_{s1\Phi}$, and $\forall q\in Q_{s2}:\forall i\in I_{s2}\cap O_{s1}:q\not\oversetmay{i}q_{s2\Phi}$.
\end{theorem}

\setcounter{theorem}{\value{theoremnumbering}}

\begin{proof}
	From Theorem~\ref{theorem:compositionality-multicast} we know that if $i_1$ and $i_2$ are strongly must-input-enabled, and $s_1$ and $s_2$ are compatible, then $\left(i_1\mathmiocophi s_1\land i_2\mathmiocophi s_2\right)\Rightarrow i_1\parallelforbidden i_2\mathmiocophi s_1\parallelforbidden s_2$.
	Additionally, $P\hidingforbidden Q=(P\parallelforbidden Q)/S$ due to Lemma~\ref{lemma:multicast-hiding-coherence} (with $P$ and $Q$ being \miaforbidden and $S=A_P\cap A_Q$).
	Furthermore, from Theorem~\ref{theorem:mioco-hiding} it follows that $i\mathmiocophi s\Rightarrow(i\mathhiding L)\mathmiocophi(s\mathhiding L)$ if $i$ is weakly must-input-enabled and $s$ is weakly may-input-enabled.
	However, \miaforbidden parallel composition with hiding does not preserve may-input-enabledness, and strong must-input-enabledness is only preserved if the automata to be composed do not contain any input transitions having a common action with the failure state as their target (\cf Corollary~\ref{corollary:hiding-input-enabledness-preservation}).
	Therefore, we require $s_1$ and $s_2$ to be strongly must-input-enabled and not containing any input transitions having a common action with the failure state as their target in order for $s_1\hidingforbidden s_2$ to be may-input-enabled.
	It follows that $\left(i_1\mathmiocophi s_1\land i_2\mathmiocophi s_2\right)\Rightarrow i_1\hidingforbidden i_2\mathmiocophi s_1\hidingforbidden s_2$ if $i_1$, $i_2$, $s_1$, and $s_2$ are strongly must-input-enabled, and $s_1$ and $s_2$ do not contain any input transitions having a common action with the failure state as their target.
\end{proof}

\subsection{Proof of Lemma~\ref{lemma:parallel-composition-hiding-associativity}}\label{subsec:parallel-composition-hiding-associativity-proof}

\setcounter{theoremnumbering}{\value{lemma}}
\setcounter{lemma}{\value{parallel-composition-hiding-associativity-counter}}

\begin{lemma}[Associativity of IR-MIA Parallel Composition]
	Let $P$, $Q$, $R$ be IR-MIA.
	It holds that $(P\parallelforbidden Q)\parallelforbidden R=P\parallelforbidden(Q\parallelforbidden R)$ if pairwise intersection of $I_P\cap O_Q$ and $I_Q\cap O_P$ with $I_Q\cap O_R$ and $I_R\cap O_Q$ results in $\emptyset$.
\end{lemma}

\setcounter{lemma}{\value{theoremnumbering}}

\begin{proof}
Correctness of Lemma~\ref{lemma:parallel-composition-hiding-associativity} directly follows from Lemma~\ref{lemma:parallel-composition-associativity} and the fact that Lemma~\ref{lemma:parallel-composition-hiding-associativity} is restricted to IR-MIA not synchronizing on the same actions. 
\end{proof}

\subsection{Proof of Theorem~\ref{theorem:decompositionality}}\label{subsec:decompositionality-proof}

\setcounter{theoremnumbering}{\value{theorem}}
\setcounter{theorem}{\value{decompositionality-counter}}

\begin{theorem}[Decompositionality of \textbf{modal-irioco}]
	Let $i$, $s$, $c_i$, and $c_s$ be \miaforbidden with
	$i$ and $c_i$ being weak must-input-enabled and all output behaviors of $i$ being mandatory.
	Then $i\mathmiocophi s$ if
	$i\sslashforbidden c_i\mathmiocophi s\sslashforbidden c_s$ and $c_i\mathmiocophi c_s$.
\end{theorem}

\setcounter{theorem}{\value{theoremnumbering}}

\begin{proof}
	For this theorem to hold, we first have to consider correctness of the
	pseudo-quotient and quotient in IR-MIA since the original operators are defined on
	the MIA model proposed by Bujtor \etal~\cite{Bujtor2015a}.
	In Def.~\ref{def:mia-transformation}, a transformation function $\mathcal{T}$ from \miaallowed to \miaforbidden is given,
	altering the semantics of the universal state (in \miaallowed)
	to become a failure state and adjusting transitions to these states, accordingly.
	To show correctness of Def.~\ref{def:pseudo-quotient} and Def.~\ref{def:quotient},
	we prove for a quotient pair $P$ and $D$ that
	$\mathcal{T}(P)\sslashforbidden\mathcal{T}(D)\cong\mathcal{T}(P\sslashallowed D)$ (\ie isomorphism).
	The isomorphism used is simply the identity function.

	Let $P$ and $D$ be \miaallowed (\ie MIA according to~\cite{Bujtor2015a}) such that $P$ and $D$ forms a quotient pair.
	The proof proceeds in two steps, (1) we show that the required property already holds for the pseudo-quotient, \ie $\mathcal{T}(P)\oslash_\Phi\mathcal{T}(D)\cong\mathcal{T}(P\oslash_A D)$ and (2) we show that the same set of states is pruned in order to obtain the quotients.

	For (1), we consider states $(p,d)$ of the pseudo quotient.
	Please note that on both sides, the state identities are preserved.
	Hence, it suffices to show that $(p,d)\oversetgamma{a} (p',d')$ is covered by both pseudo-quotient operations.
	For rules (QMay1) to (QMay3) and (QMust1) to (QMust3), this obviously holds.
	The only difference is the handling of the universal/failure state by the
	remaining rules, (QMay4) and (QMay5) in \cite{Bujtor2015a} and (QMay4/QMust4) in Def.~\ref{def:pseudo-quotient}.

	If $(p,d)\oversetmaysub{a}{P\oslash_A D} (e_P,e_D)$ (\ie $(p,d)\not\oversetmaysub{a}{\mathcal{T}(P\oslash_A D)}$) due to (QMay4), $p\oversetmaysub{a}{P} e_P$ and by $\mathcal{T}$, $p\not\oversetmaysub{a}{\mathcal{T}(P)}$.
	Thus, $(p,d)\not\oversetmaysub{a}{\mathcal{T}(P)\oslash_\Phi\mathcal{T}(D)}$.
	If $(p,d)\oversetmaysub{a}{P\oslash_A D} (e_P,e_D)$ due to (QMay5), $p\neq e_P$, $d\not\oversetmaysub{a}{D}$ and $a\in A_D\setminus (O_P\cap I_D)$.
	Again, neither of the rules of Def.~\ref{def:pseudo-quotient} applies.
	Thus, $(p,d)\not\oversetmaysub{a}{\mathcal{T}(P)\oslash_\Phi\mathcal{T}(D)}$.
	Suppose $(p,d)\oversetmustsub{a}{\mathcal{T}(P)\oslash_\Phi \mathcal{T}(D)} (p_\Phi,d_\Phi)$ due to rule (QMay4/QMust4) of Def.~\ref{def:pseudo-quotient}.
	In this case, $p\oversetmustsub{a}{\mathcal{T}(P)} p_\Phi$, \ie $p\not\oversetmaysub{a}{P}$.
	Since none of the rules (QMay4) or (QMay5) in \cite{Bujtor2015a} applies, $(p,d)\oversetmaysub{a}{P\oslash_A D}$ thus $(p,d)\oversetmaysub{a}{\mathcal{T}(P\oslash_A D)}$

	For step (2), we have to consider the rules to identify impossible states.
	The rules of Def.~\ref{def:quotient} are adopted from Bujtor \etal~\cite{Bujtor2015a}.
	The one rule missing in our set is due to the fact that we do not have a universal state in \miaforbidden.

	Therefore, $\mathcal{T}(P)\sslashforbidden\mathcal{T}(D)\cong\mathcal{T}(P\sslashallowed D)$.
	We now proceed the proof of Theorem~\ref{theorem:decompositionality}.

	Let $i'=(i\sslashforbidden c_i)\parallelforbidden c_i$ and $s'=(s\sslashforbidden c_s)\parallelforbidden c_s$.
	From~\cite{Bujtor2015a} and the first part of this proof, we conclude that $i'\miarefforbidden i$ and $s'\miarefforbidden s$.
	Therefore, we have to show that $i'\mathmiocophi s'\Rightarrow i\mathmiocophi s$.
	\miaforbidden $i$ is weak must-input-enabled, and, therefore, $i'$ is also weak must-input-enabled, \ie $i$ and $i'$ have the same input behavior.
	\miaforbidden $s$ may only have less than, or equal input behaviors as $s'$ as, under refinement, it is only possible to add inputs but not to remove inputs.
	Therefore, we only have to consider output behaviors.
	Hence, $i$ and $i'$ only differ in output behaviors, \ie mandatory outputs and forbidden outputs of $i'$ may be optional in $i$.
	However, both cases are not possible as we restrict $i$ to only have mandatory outputs.
\end{proof}

\subsection{Proof of Theorem~\ref{theorem:and-is-and}}\label{subsec:and-is-and-proof}

\setcounter{theoremnumbering}{\value{theorem}}
\setcounter{theorem}{\value{and-is-and-counter}}

\begin{theorem}[$\land$ is And]
	Let $p$ and $q$ be \miaforbidden{} with common alphabets.
	Then, (1) $(\exists r:r\miarefforbidden p\text{ and } r\miarefforbidden q)$ iff $p\land q$ is defined.
	Further, in case $p\land q$ is defined: (2) $r\miarefforbidden p$ and $r\miarefforbidden q$ iff $r\miarefforbidden p\land q$. 
\end{theorem}

\setcounter{theorem}{\value{theoremnumbering}}

\begin{proof}
	Let $p$ and $q$ be \miaforbidden{} with common alphabets.
	We first prove (2), that is if $p\land q$ is defined, then for every \miaforbidden{} $r$ it holds that $r\miarefforbidden p$ and $r\miarefforbidden q$ iff $r\miarefforbidden p\land q$.
	\begin{description}
		\item[if:] Let $r$ be a \miaforbidden{} such that $r\miarefforbidden p\land q$ by refinement relation $\mathcal R$.
		We define\\[.3em] 
		\centerline{$\mathcal R_p := \{ (r_1,p_1) \mid (r_1,(p_1,q_1))\in\mathcal R \} \cup \{ (r_\Phi,p_\Phi) \}$}\\[.3em]
		and prove that $\mathcal R_p$ is a refinement relation proving $r\miarefforbidden p$ ($\mathcal R_q$ for $r\miarefforbidden q$, analogously).
		We proceed by the steps of Def.~\ref{def:mia-refinement-forbidden} for every $(r_1,p_1)\in\mathcal R_p$ with $r_1\neq r_\Phi$.
		\begin{enumerate}
			\item We need to show that $p_1\neq p_\Phi$.
			By construction of $p\land q$ (Def.~\ref{def:conjunctive-product}), there is no combined state $(p_\Phi,q')$ (by IMust1, IMust2, IMay, DMay1, FMay2, FMust).
			Thus, $\mathcal R$ cannot relate $r_1$ with such a state, as long as $r_1\neq r_\Phi$.
			
			\item Suppose $p_1 \oversetmust{i} p_2 \neq p_\Phi$.
			Since $(r_1,p_1)\in\mathcal R_p$, there is a $q_1$ such that $(r_1,(p_1,q_1))\in\mathcal R$.
			Then, by IMust1 or DMay1, there exists some $q_2$ such that $(p_1,q_1)\oversetmust{i} (p_2,q_2)$ in the conjunctive product $p\& q$.
			This transition is not pruned by constructing $p\land q$, since otherwise, the pair $(p_1,q_1)$ is also pruned, contradicting the assumption that $\mathcal R$ is a refinement relation.
			Thus, there exists an $r_2\neq r_\Phi$ such that $r_1 \oversetmust{i}\Oversetmust{\varepsilon} r_2$ and $(r_2,(p_2,q_2))\in\mathcal R$.
			By definition of $\mathcal R_p$, $(r_2,p_2)\in\mathcal R_p$.
			
			\item The case $p_1 \oversetmust{\omega} p_2$ is analogous, using OMust1 to construct a respective $r_2$ with $(r_2,p_2)\in\mathcal R_p$.
			
			\item Suppose $r_1 \oversetmay{i} r_2$ and $p_1\oversetmay{i}$.
			We need to show that there is a $p_2$ with $p_1\oversetmay{i}\Oversetmay{\varepsilon} p_2$ such that $(r_2,p_2)\in\mathcal R_p$.
			By $r\miarefforbidden p\land q$ and the construction of $\mathcal R_p$, there is a $q_1$ with $(r,(p_1,q_1))\in\mathcal R$ and $(p_1,q_1)\oversetmay{i}\Oversetmay{\varepsilon} (p_2,q_2)$ with $(r_2,(p_2,q_2))\in\mathcal R$, only if $(p_1,q_1)\oversetmay{i}$.
			If so, take the respective $p_2$ and $(r_2,p_2)\in \mathcal R_p$.
			$(p_1,q_1)\oversetmay{i}$ always holds, since $p_1\oversetmay{i}$ by assumption and rule (IMay), (DMay1), or (FMay) ensure its existence.
			Note that the latter rule covers the case where $(p_1,q_1)\oversetmay{i} (p\land q)_\Phi$ (forbidden state of the conjunction).
			In this case, we get $p_2 = p_\Phi$ and $r_2= r_\Phi$.
			
			\item The case $p_1 \oversetmay{i} p_2$ where $p_2\neq p_\Phi$ is analogous to case 2, using IMay to derive an $r_2$ such that $r_1 \oversetmay{i}\Oversetmay{\varepsilon} r_2$ and $(r_2,p_2)\in\mathcal R_p$.
			Consider $p_2=p_\Phi$.
			$p_1 \oversetmay{i} p_2$ is complemented by $q_1$ using rule (FMay).
			Thus, we get $(p_1,q_1)\oversetmay{i} (p\land q)_\Phi$.
			By $r\miarefforbidden p\land q$ (condition (5) of Def.~\ref{def:mia-refinement-forbidden}), then $r_1 \oversetmay{i} r_\Phi$.
			Thus, $(r_\Phi,p_2)=(r_\Phi,p_\Phi)\in\mathcal R_p$.
			
			\item Suppose $r_1 \oversetmay{\omega} r_2$.
			From $\mathcal R$, we obtain a $q_1$ such that $(r_1,(p_1,q_1))\in\mathcal R$.
			Since $r\miarefforbidden p\land q$, there is a transition $(p_1,q_1)\Oversetmay{\hat\omega} (p_2,q_2)$ such that $(r_2,(p_2,q_2))\in\mathcal R$.
			Transition $(p_1,q_1)\Oversetmay{\hat\omega} (p_2,q_2)$ is due to repeated application of (OMay).
			Thus, $p_1 \Oversetmay{\hat\omega} p_2$ and by construction of $\mathcal R_p$, $(r_2,p_2)\in\mathcal R_p$.
		\end{enumerate}
		\item[only if:]
		Suppose we have $r\miarefforbidden p$ and $r\miarefforbidden q$ witnessed by refinements $\mathcal R_p$ and $\mathcal R_q$.
		We define\\[.3em]
		\centerline{$\mathcal R := \{ (r_1,(p_1,q_1)) \mid (r_1,p_1)\in\mathcal R_p \land (r_1, q_1) \in R_q \} \cup \{ (r_\Phi,(p\land q)_\Phi \}$}\\[.3em]
		and show that $\mathcal R$ is a refinement relation.
		We follow the steps in Def.~\ref{def:mia-refinement-forbidden} for all $(r_1,(p_1,q_1))\in\mathcal R$ such that $r_1\neq r_\Phi$:
		\begin{enumerate}
			\item This case ($(p_1,q_1)\neq(p\land q)_\Phi$) holds by construction of $\mathcal R$.
			
			\item In this case, we have $(p_1,q_1)\oversetmust{i} (p_2,q_2)$ due to (IMust1) or (IMust2).
			W.\,l.\,o.\,g., we consider (IMust1), \ie $p_1\oversetmust{i} p_2\neq p_\Phi$ and $q_1 \oversetmay{i}\Oversetmay{\varepsilon} q_2\neq q_\Phi$.
			Furthermore, note that $p_2\land q_2$ is defined, \ie not pruned.
			By $(r_1, p_1)\in\mathcal R_p$, we have an $r_2$ such that $r_1\oversetmust{i}\Oversetmust{\varepsilon} r_2\neq r_\Phi$ and $(r_2,p_2)\in\mathcal R_p$.
			Since $r_1 \oversetmay{i}\Oversetmay{\varepsilon} r_2$ and $(r_1, q_1)\in\mathcal R_q$, there are $\hat{q_2}$ such that $q_1 \oversetmay{i}\Oversetmay{\varepsilon} \hat{q_2}$ with $(r_2, \hat{q_2})\in\mathcal R_q$.
			Since $p_2\land q_2$ is defined, there is at most one such $\hat{q_2}$, \ie $q_2 = \hat{q_2}$.
			Thus, $(r_2, q_2)\in\mathcal R_q$ and $(r_2,(p_2,q_2))\in\mathcal R$.
			
			\item This case is analogous to the one before, using (OMust1) und (OMust2).
			
			\item This case is analogous to case 6, using (IMay), (DMay1), (DMay2) and (FMay).
			
			\item Here, we distinguish between (a) $(p_1,q_1)\oversetmay{i} (p_2,q_2)$ due to rules (IMay), (DMay1), or (DMay2), and (b) $(p_1,q_1)\oversetmay{i} (p\land q)_\Phi$ due to (FMay).
			In case (b), it holds that $p_1 \oversetmay{i} p_\Phi$ or $q_1 \oversetmay{i} q_\Phi$.
			Since $r\miarefforbidden p$ and $r\miarefforbidden q$, we may deduce $r_1 \oversetmay{i} r_\Phi$, completing the case.
			In case (a), we further distinguish by the respective rules:
			\begin{description}
				\item[(IMay)] Here, we get that $p_1 \oversetmay{i}\Oversetmay{\varepsilon} p_2 \neq p_\Phi$ and $q_1 \oversetmay{i}\Oversetmay{\varepsilon} q_2 \neq q_\Phi$.
				By assumption, $r\miarefforbidden p$ and $r\miarefforbidden q$.
				Thus, $r_1 \oversetmay{i}\Oversetmay{\varepsilon}$ with appropriate target states of $p_2$ and $q_2$.
				If there is no target state $r_2$ such that $(r_2,p_2)\in\mathcal R_p$ and $(r_2,q_2)\in\mathcal R_q$, there {\em must} be a contradiction following $p_2$ and $q_2$, making their conjunction $p_2\land q_2$ undefined.
				But $p_2\land q_2$ is defined.
				Thus, such an $r_2$ exists.
				
				\item[(DMay1)] In this case, we get that $q_2 = q_d$ and for every $r_i$, it holds that $(r_i,q_d)\in\mathcal R_q$.
				
				\item[(DMay2)] Completely analogous. 
			\end{description}
			
			\item In case we have $r_1 \oversetmay{\omega} r_2$
			Since $r\miarefforbidden p$ and $r\miarefforbidden q$, we have that $p_1 \Oversetmay{\hat\omega} p_2$ and $q_1 \Oversetmay{\hat\omega} q_2$ such that $(r_2,p_2)\in\mathcal R_p$ and $(r_2,q_2)\in\mathcal R_q$.
			Thus $(r_2,(p_2,q_2))\in\mathcal R$.
			By repeated application of rule (OMay), we get that $(p_1,q_1) \Oversetmay{\hat\omega} (p_2,q_2)$.
		\end{enumerate}
	\end{description}
	Part (1) follows from (2) and the fact that $\miarefforbidden$ is reflexive, \ie if $p\land q$ is defined, then there is an $r$ such that $r\miarefforbidden p$ and $r\miarefforbidden q$.
	Since $p\land q \miarefforbidden p\land q$, (2) allows to deduce $p\land q \miarefforbidden p$ and $p\land q \miarefforbidden q$.
\end{proof}

\subsection{Proof of Lemma~\ref{lemma:conjunction-associativity}}\label{subsec:conjunction-associativity-proof}

\setcounter{theoremnumbering}{\value{lemma}}
\setcounter{lemma}{\value{conjunction-associativity-counter}}

\begin{lemma}[Associativity of IR-MIA Conjunction]
	Let $P$, $Q$, $R$ be IR-MIA.
	Then, (1) $P\land (Q\land R)$ is defined iff $(P\land Q)\land R$ is defined, and (2) if $P\land (Q\land R)$ is defined, then $S\miarefforbidden P\land (Q\land R)$ iff $S\miarefforbidden (P\land Q)\land R$.
\end{lemma}

\setcounter{lemma}{\value{theoremnumbering}}

\begin{proof}
	(1) Theorem~\ref{theorem:and-is-and} parts~(1) and~(2) imply that $P\land (Q\land R)$ is defined iff $\exists S.S\miarefforbidden P$ and $S\miarefforbidden Q\land R$ iff $\exists S.S\miarefforbidden P$ and $S\miarefforbidden Q$ and $S\miarefforbidden R$ iff $\exists S.S\miarefforbidden P\land Q$ and $S\miarefforbidden R$ iff $(P\land Q)\land R$ is defined.
	Statement~(2) follows directly from multiple applications of Theorem~\ref{theorem:and-is-and} part~(2).
\end{proof}

\subsection{Proof of Theorem~\ref{theorem:conjunction}}\label{subsec:conjunction-proof}

\setcounter{theoremnumbering}{\value{theorem}}
\setcounter{theorem}{\value{conjunction-counter}}

\begin{theorem}[Compositionality of Conjunction of \textbf{modal-irioco}]
	Let $s$, $s'$ and $i$ be \miaforbidden.
	Then it holds that $\left(i\mathmiocophi s\land i\mathmiocophi s'\right)\Rightarrow i\mathmiocophi s\land s'$.
\end{theorem}

\setcounter{theorem}{\value{theoremnumbering}}

\begin{proof}
	Let $s,s',i$  be \miaforbidden such that $s\land s'$ is defined, $i\mathmiocophi s$, and $i\mathmiocophi s'$.
	We need to show that
	\begin{enumerate}
		\item for all $\sigma\in\maystraces(s\land s')$, $\mayout(i\mayafter \sigma)\subseteq \mayout(s\land s' \mayafter \sigma)$, and
		\item for all $\sigma\in\maystraces(i)$, $\mustout(s\land s' \mayafter \sigma)\subseteq\mustout(i\mayafter \sigma)$.
	\end{enumerate}
	For case 1, let $\sigma\in\maystraces(s\land s')$.
	By construction of $s\land s'$ using rules (OMay), (IMay), (DMay1), (DMay2), and (FMay), $\sigma\in\maystraces(s)$ or $\sigma\in\maystraces(s')$.
	In case both happen to be true, $\mayout(s\land s' \mayafter \sigma) = \mayout(s\mayafter \sigma)\cap\mayout(s'\mayafter \sigma)$ by (OMay).
	By assumptions $i\mathmiocophi s$ and\linebreak$i\mathmiocophi s'$, we get $\mayout(i\mayafter \sigma)\subseteq\mayout(s\mayafter\sigma)\cap\mayout(s'\mayafter\sigma)=\mayout(s\land s'\mayafter\sigma)$.
	In case $\sigma\in\mayout(s\mayafter\sigma)\setminus\mayout(s'\mayafter\sigma)$ (analogous argumentation for the symmetric case),\linebreak$\mayout(s\land s'\mayafter\sigma)=\mayout(s\mayafter\sigma)$ and at some point in $s'$, the continuation of a prefix of $\sigma$ is underspecified.
	Thus, we deduce $\mayout(i\mayafter\sigma)\subseteq\mayout(s\mayafter\sigma)=\mayout(s\land s'\mayafter\sigma)$.

	In case 2, let $\sigma\in\maystraces(i)$.
	It suffices to consider the case of $s\land s'\mayafter\sigma \neq\emptyset$.
	Thus, as above $\sigma\in\maystraces(s)$ or $\sigma\in\maystraces(s')$.
	It holds that $\mustout(s\land s'\mayafter\sigma)\subseteq\mustout(s\mayafter\sigma)\cup\mustout(s'\mayafter\sigma)$.
	Since $i\mathmiocophi s$ and $i\mathmiocophi s'$, $\mustout(s\land s'\mayafter\sigma)\subseteq\mustout(i\mayafter\sigma)$.
\end{proof}

\end{appendix}

\end{document}